\documentclass{article}
\usepackage{graphicx}
\RequirePackage[OT1]{fontenc}
\RequirePackage{amsthm,amsmath}
\RequirePackage[numbers]{natbib}
\RequirePackage[colorlinks,citecolor=blue,urlcolor=blue]{hyperref}
\usepackage[usenames]{color}
\usepackage{lscape}

\numberwithin{equation}{section}
\theoremstyle{plain}

\newtheorem{assumption}{Assumption}[section]
\newtheorem{lemma}[assumption]{Lemma}
\newtheorem{definition}[assumption]{Definition}
\newtheorem{proposition}[assumption]{Proposition}
\newtheorem{corollary}[assumption]{Corollary}
\newtheorem{remark}[assumption]{Remark}

\newtheorem{example}[assumption]{Example}

\begin{document}

\title{Local degeneracy of Markov chain Monte Carlo methods}
\date{}
\author{Kengo KAMATANI\footnote{
Graduate School of Engineering Science, Osaka University, 
Machikaneyama-cho 1-3, Toyonaka-si, Osaka, 560-0043 , Japan, kamatani@sigmath.es.osaka-u.ac.jp. }
\footnote{Supported in part by Grant-in-Aid for Young Scientists (B) 22740055.}
}

\maketitle

\begin{abstract}
We study asymptotic behavior of Monte Carlo 
method. Local consistency is one of an ideal property of 
Monte Carlo method. However, it may fail to hold local consistency for 
several reason.
In fact, in practice, it is more important to study such a non-ideal behavior.
We call local degeneracy for one of a non-ideal behavior of 
Monte Carlo methods. We show some equivalent conditions for 
local degeneracy. 
As an application we study a Gibbs sampler (data augmentation) for
cumulative logit model with or without marginal augmentation. 
It is well known that natural Gibbs sampler does not work well for this model. 
In a sense of local consistency and degeneracy, marginal augmentation 
is shown to improve the asymptotic property. 
However, when the number of categories is large, 
both methods are not locally consistent. 
\end{abstract}

\section{Introduction}

This paper investigates a poor behavior of Markov chain Monte Carlo (MCMC) method.
There have a vast literature related to the sufficient conditions for a good behavior, ergodicity:
see reviews \cite{TierneyAOS94} and \cite{RR}
and textbooks such as \cite{N} and \cite{MT}. 
The Markov probability transition kernel of MCMC is Harris recurrent 
under fairly general assumptions. Moreover, 
it is sometimes geometrically ergodic. 
In practice, 
however  the performance can be bad even if it is geometrically 
ergodic. 


 In \cite{Kamatani10} we introduced a framework for the analysis of 
 Monte Carlo procedure. 
 Monte Carlo procedure is defined as a pair $\mathcal{M}=(M,e)$
 of underlying probability structure $M$ and a sequence of
 ``estimator'' $e=(e_m;m=1,2,\ldots)$ for the target probability distribution. 
Using the framework we constructed consistency, which is a good behavior of 
Monte Carlo procedure. 
Current study, we apply the framework to study bad behavior.

There are several bad behaviors for Monte Carlo procedure. 
Two extreme cases are, a) the sequence generated by Monte Carlo 
procedure has very poor mixing property, and b)
the sequence goes out to infinity. 
We call a) degeneracy and the paper is devoted to the study of 
the property. 
We focus on a) in this paper. 
For b), see Examples 3.2 and 3.3 of \cite{Kamatani10}.

\subsection{Degeneracy}

To describe degeneracy more precisely, we consider a numerical simulation for the following 
simple model:
\begin{displaymath}
P(Y=1|\theta,x)=\Phi(\theta x),\ 
P(Y=0|\theta,x)=1-P(Y=1|\theta,x)
\end{displaymath}
where $x$ is a $\mathbf{R}$-valued explanatory variable and 
$\theta$ is a parameter
and $\Phi$ is a cumulative distribution function of the normal distribution (See Section \ref{normal}). 
Explanatory variable $x$ is generated from uniformly distribution on $(0,1)$.
We define two Gibbs sampler $\mathcal{M}_n$ and $\mathcal{N}_n$.

\subsubsection{Gibbs sampler $\mathcal{M}_n$}
Assume we have observation $y_n=(y^1,\ldots, y^n)$ and 
$x_n=(x^1,\ldots, x^n)$ and  $\theta$ prior is set to be standard normal distribution. 
There are two ways for construction of the Gibbs sampler. 
One way is to prepare latent variable $z\in\mathbf{R}$ from 
$N(0,1)$  and set
\begin{displaymath}
 y=\left\{\begin{array}{cc}1&\mathrm{if}\ z\le\theta x \\ 0& \mathrm{if}\ z>\theta x\end{array}\right. .
\end{displaymath}
Then Gibbs sampler is generated by iterating the following procedure:
For given $\theta$ and for $i=1,\ldots, n$, generate
$z^i$ from
$N(0,1)$
 truncated to $(-\infty,\theta x^i]$
if $y^i=1$
and
truncated to
 $(\theta x^i,\infty)$  if $y^i=0$.
Then update $\theta$ from $N(0,1)$ truncated to 
an interval
\begin{displaymath}
[\max_{y^i=1}\frac{z^i}{x^i}.\min_{y^i=0} \frac{z^i}{x^i}).
\end{displaymath} 
Write $\mathcal{M}_n$ for this Gibbs sampler. 

\subsubsection{Gibbs sampler $\mathcal{N}_n$}
Similarly, we define another Gibbs sampler by taking latent variable $z^i$
from $N(-\theta x^i,1)$, which is a  normal distribution with mean 
$-\theta x^i$ with variance $1$ and set 
\begin{displaymath}
 y=\left\{\begin{array}{cc}1&\mathrm{if}\ z\le0 \\ 0& \mathrm{if}\ z>0\end{array}\right. .
\end{displaymath}
Then Gibbs sampler is generated by iterating the following procedure:
For given $\theta$ and for $i=1,\ldots, n$, generate
$z^i$ from
$N(-\theta x^i,1)$
 truncated to $(-\infty,0]$
if $y^i=1$
and  truncated to $(0,\infty)$
if $y^i=0$.
Then update $\theta$ from normal distribution with mean $\mu$ and 
variance $\sigma^2$ defined by
\begin{displaymath}
\mu=-\frac{\sum_{i=1}^n x^iz^i}{1+\sum_{i=1}^n (x^i)^2},\ 
\sigma^2=\frac{1}{1+\sum_{i=1}^n (x^i)^2}. 
\end{displaymath}
Write $\mathcal{N}_n$ for this Gibbs sampler. 

We obtain two Gibbs samplers $\mathcal{M}_n$ and $\mathcal{N}_n$. 
Although the constructions are similar
and both of which have geometric ergodicity, the performances are different. 
\begin{figure}[htbp]
\includegraphics[width=12cm,bb=0 0 779 500]{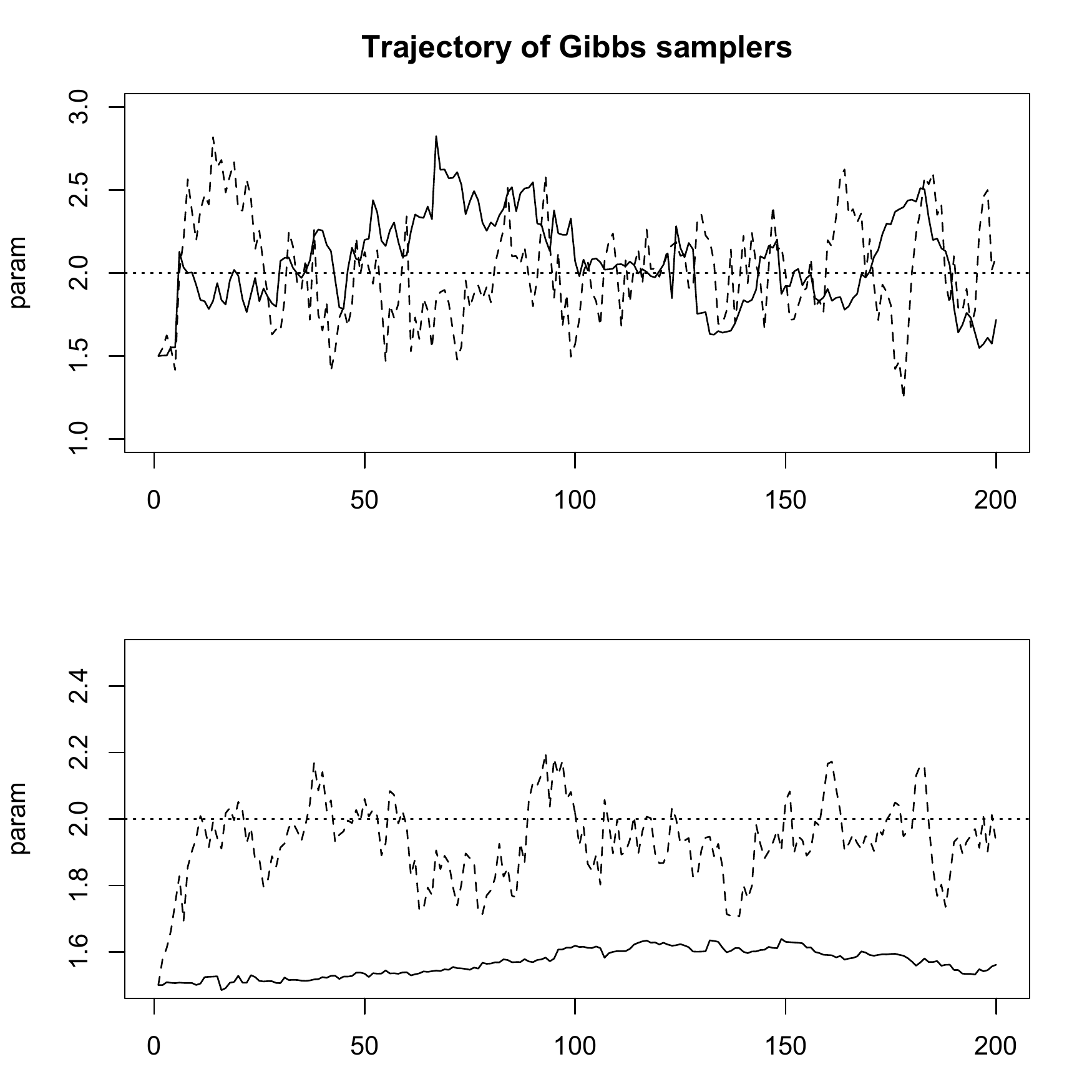}
\caption{Trajectory of the Gibbs samplers for sample size $n=100$ (upper) and $n=1000$ (lower). Solid line is for $\mathcal{M}_n$ and dashed lines is for $\mathcal{N}_n$.\label{Figure1}}
 \end{figure}
Figure \ref{Figure1} is a trajectory of the Gibbs sampler sequence 
\begin{displaymath}
\theta(0),\ldots, \theta(m-1)
\end{displaymath}
for iteration $m=200$ and sample size $n=100$ (upper) and  $n=1000$ (lower). 
For each sample size, by ergodicity, empirical distributions tend to the same posterior distribution of $\theta$ for two Gibbs samplers as $m\rightarrow\infty$. 
However the solid line $\mathcal{M}_n$ has poor mixing property than 
$\mathcal{N}_n$. Therefore it may produces a poor estimation of the posterior distribution. 

The difference becomes larger when the sample size $n=1000$
in Figure \ref{Figure1} (lower). 
The trajectory from $\mathcal{M}_n$ (solid line)
is almost constant. 
For both simulations, the true value is $\theta_0=2$
and the initial value $\theta(0)$ is set to $1.5$. 

Even though $\mathcal{M}_n$ has geometric ergodicity, it has poor mixing property. 
We would like to say $\{\mathcal{M}_n;n=1,2,\ldots\}$ is degenerate. Later we will prove that it is degenerate after certain localization. 
On the other hand $\{\mathcal{N}_n;n=1,2,\ldots\}$ is consistent under the same scaling by
Theorem 6.4 of  \cite{Kamatani10}. 

We study such a poor behavior, degeneracy, in this paper. 
The analysis may seem to be just a formalization of obvious facts. 
However, sometimes degeneracy can not be directly visible
and it produces non-intuitive results.  
In this paper, we obtain the following results
for (Markov chain) Monte Carlo methods.

\begin{enumerate}
\item
Degeneracy and local degeneracy of Monte Carlo procedure are defined and analyzed. 
\item 
As an example, we studied cumulative link model. Marginal augmentation method 
is known to work at least as good as the original Gibbs sampler. 
We show that in some cases, marginal augmentation really improves the asymptotic property, and the rest of the cases, surprisingly, we show that both of the MCMC methods does not have 
local consistency. 
\end{enumerate}


The paper is organized as follows. 
Section \ref{dmc} is devoted to a study of degeneracy of Monte Carlo procedure in general. 
In Subsection \ref{clc} we briefly review consistency of Monte Carlo procedure, and after that we define degeneracy and apply it to Markov chain Monte Carlo procedure. 
Next we examine the degeneracy for an example, cumulative link model. 
We prepare Section \ref{apc} for the asymptotic property of cumulative link model
itself. There is no Monte Carlo procedure in this section. 
In Section \ref{ags} we apply degeneracy to the model and obtain asymptotic properties of Markov chain Monte Carlo methods for cumulative link model.

\subsection{Notation}
Let $\mathbf{N}=\{1,2,\ldots,\}$ and $\mathbf{N}_0=\{0,1,2,\ldots\}$. 
We write the integer part of $x\in\mathbf{R}$ by $[x]$. 

\subsubsection{Probability measure, Transition kernel}\label{tk}
For measurable spaces $(E,\mathcal{E})$, 
the space of probability measures on $(E,\mathcal{E})$
is denoted by $\mathcal{P}(E)$. 

For two measurable space $(E,\mathcal{E})$ and $(F,\mathcal{F})$, 
a probability transition kernel $K$ from $E$ to $F$ is 
a map $K:E\times\mathcal{F}\rightarrow [0,1]$ such that
\begin{enumerate}
\item $K(x,\cdot)$ is a probability measure on $(F,\mathcal{F})$ for 
$x\in E$. 
\item $K(\cdot, A)$ is $\mathcal{E}$-measurable for any $A\in\mathcal{F}$. 
\end{enumerate}
We may write $K(dy|x)$ instead of $K(x,dy)$. 
If $K(x,\cdot)$ is $\sigma$-finite measure instead of probability measure, 
we call $K$ a transition kernel. 

\subsubsection{Normal distribution}\label{normal}
Write $\phi(x)=\exp(-x^2/2)/\sqrt{2\pi}$
for a probability distribution function of $N(0,1)$
and write $\Phi(x)=\int_{-\infty}^x \phi(y)dy$. 
For $\mu\in\mathbf{R}^p$ and $p\times p$-positive definite matrix $\Sigma$, 
a function 
$\phi(x;\mu,\Sigma)=\exp(-x^T\Sigma^{-1}x/2)/(2\pi\mathrm{det}(\Sigma))^{1/2}$
is a probability distribution function of $N(\mu,\Sigma)=N_p(\mu,\Sigma)$
where $\mathrm{det}(\Sigma)$ is a determinant of $\Sigma$ and 
$x^T$ is a transpose of a vector $x\in\mathbf{R}^p$. 

\subsubsection{Central value}
For a probability measure $\mu$ on $\mathbf{R}$, a central value is 
a point $\overline{x}\in\mathbf{R}$ satisfying 
\begin{displaymath}
\int_{\mathbf{R}} \arctan(x-\overline{x})\mu(dx)=0.
\end{displaymath}
Element of $\mathbf{R}^p$ is denoted by $x=(x^1,\ldots, x^p)^T$. 
For a probability measure $\mu$ on  $\mathbf{R}^p$, 
let $\mu^i(A)$ be $\int_{x\in\mathbf{R}} 1_A(x^i)\mu(dx)$ for $A\in\mathcal{B}(\mathbf{R})$. 
For $\mu$, 
we call
$\overline{x}=(\overline{x}^1,\overline{x}^2,\ldots,\overline{x}^p)^T\in\mathbf{R}^p$
 central value if each $\overline{x}^i$ is a central value of $\mu^i$. 
There is no practical reason for the use of the central value for 
Markov chain Monte Carlo procedure as is used in this paper. 
We use it because of its existence and continuity. That is, (a) for the posterior distribution
  $P_n(d\theta|x_n)$, its mean does not always exist but the central value does and moreover, it is unique and (b) if $\mu_n\rightarrow \mu$, then the central value of $\mu_n$
  tends to that of $\mu$. See \cite{Ito}. 

\section{Degeneracy of Markov chain Monte Carlo procedure}\label{dmc}

In this section, we introduce a notion of degeneracy and local degeneracy of 
Monte Carlo procedure. We use the same framework as \cite{Kamatani10}
to describe local degeneracy. In their approach, Monte Carlo procedure 
is considered to be a pair of random probability measure and transition kernels. 
We briefly review their framework in Subsection \ref{clc}.

\subsection{Consistency and local consistency}\label{clc}

In this subsection, we prepare a quick review of the framework of \cite{Kamatani10}. 
Let  $(S,\mathcal{S})$ be a measurable space. 
Let $(S^{\mathbf{N}_0},\mathcal{S}^{\mathbf{N}_0})$ be a countable product of $(S,\mathcal{S})$. 
Each element of $S^{\mathbf{N}_0}$ is denoted by 
$s_\infty=(s(0),s(1),\ldots)$ and its first $m$ subsequence is denoted by 
$s_m=(s(0),\ldots, s(m-1))$. 
Let $(\Theta,d)$ be a complete separable metric space equipped with Borel $\sigma$-algebra $\Xi$. 
We define non-random Monte Carlo procedure. The meaning of ``non-random'' will be clear after we define ``random'' Mote Carlo procedure
in Definition \ref{mcp}
and standard Gibbs sampler in Definition \ref{ssg}. 

\begin{definition}[Non-random Monte Carlo procedure]
A pair $\mathcal{M}=(M,e)$ is said to be non-random Monte Carlo procedure on $(S,\Theta)$ where 
$M$ is a probability measure on $S^{\mathbf{N}_0}$ and $e=(e_m;m=1,2,\ldots)$ is a sequence of probability transition kernels
$e_m$ from
$S^m$ to $\Theta$. 
\end{definition}

A simplest example of non-random Monte Carlo procedure is a non-random crude Monte Carlo procedure.

\begin{example}[Crude Monte Carlo]
If we want to calculate an integral $\int_\Theta f(\theta)\Pi(d\theta)$
for probability measure $\Pi$ and measurable function $f$, one approach 
is to generate i.i.d. sequence $\theta(0),\theta(1),\ldots$ from $\Pi$
and calculate $m^{-1}\sum_{i=0}^{m-1}f(\theta(i))$. 
In this case $S=\Theta$ and we write $\theta_m$ and $\theta_\infty$ instead of $s_m$ and $s_\infty$. 
This simple Monte Carlo method is sometimes called a crude Monte Carlo method.
We can describe it as a non-random Monte Carlo procedure. Let $M$ be a countable product of a probability measure $\Pi$ on $\Theta$
(that is, $M=\Pi^{\otimes \mathbf{N}_0}$)
and  $e_m$ be
\begin{equation}\label{emp}
e_m(\theta_m,\cdot)=\frac{1}{m}\sum_{i=0}^{m-1}\delta_{\theta(i)}\ (\theta_m=(\theta(0),\ldots,\theta(m-1)))
\end{equation}
where $\delta_\theta$ is a Dirac measure. Then 
$\int_\Theta f(\theta)e_m(\theta_m,d\theta)=m^{-1}\sum_{i=0}^{m-1}f(\theta(i))$. 
We call $(M,e)$ a crude Monte Carlo procedure.
\end{example}

\begin{example}[Accept-Reject method]
Accept-reject method generate i.i.d. sequence from $\Pi$
on $\Theta$ from another probability measure $Q$. 
Assume that $\Pi$ is absolutely continuous with respect to $Q$
and for some $M<\infty$, 
\begin{displaymath}
r(\theta):=\frac{d\Pi}{dQ}(\theta)\le M\ (\theta\in\Theta). 
\end{displaymath}
Generate i.i.d. sequence $\theta(0),\theta(1),\ldots$ from $Q$
and $u(0),u(1),\ldots$ from the uniform distribution $U[0,1]$. 
Then accept-reject method approximate $\Pi$ by
\begin{equation}\label{ar}
\sum_{i=0}^{m-1}\delta_{\theta(i)}\frac{1(u(i)\le M^{-1}r(\theta(i)))}{\sum_{i=0}^{m-1}1(u(i)\le M^{-1}r(\theta(i)))}.
\end{equation}
We can describe it as a non-random Monte Carlo procedure. Let $M$ be a countable product of a probability measure $Q\otimes U$ on $S:=\Theta\times [0,1]$
(that is, $M=(Q\otimes U)^{\otimes \mathbf{N}_0}$)
and  $e_m(s_m,\cdot)$ be as (\ref{ar})
where $s_m=(s(0),s(1),\ldots, s(m-1))$ and $s(i)=(\theta(i),u(i))\in S$. 
We call $(M,e)$ accept-reject procedure for $e=(e_m;m=1,2,\ldots)$. 
\end{example}

Now we consider a random Monte Carlo procedure. Let $(X,\mathcal{X},P)$ be a probability space. 
\begin{definition}[Monte Carlo procedure]\label{mcp}
A pair $\mathcal{M}=(M,e)$ is said to be Monte Carlo procedure defined on $(X,\mathcal{X},P)$ on $(S,\Theta)$ where 
$M$ is a probability transition kenel from $X$ to $S^{\mathbf{N}_0}$, that is
\begin{enumerate}
\item $M(x,\cdot)$ is a probability measure on $(S^{\mathbf{N}_0},\mathcal{S}^{\mathbf{N}_0})$. 
\item $M(\cdot,A_\infty)$ is $\mathcal{X}$-measurable for any $A_\infty\in \mathcal{S}^{\mathbf{N}_0}$. 
\end{enumerate}
and $e=(e_m;m=1,2,\ldots)$ is a sequence of a probability transition kernel 
$e_m$ from 
$X\times S^m$ to $\Theta$.
\end{definition}

We call $\mathcal{M}$ stationary if $M(x,\cdot)$ is (strictly) stationary for $P$-a.s. $x$. Stationarity plays an important role for the asymptotic behavior of Monte Carlo procedure.

Markov chain Monte Carlo procedure is a class of Monte Carlo procedure. 
Let $\mu$ be a probability transition kernel from $X$ to $S$ and $K$ be a probability transition kernel 
from $X\times S$ to $S$. We call a probability transition kernel$M$ from $X$ to $S^{\mathbf{N}_0}$
random Markov measure generated by $(\mu,K)$ if $M(x,\cdot)$ is a Markov measure having initial probability distribution 
$\mu(x,\cdot)$ and a probability transition kernel $K(x,\cdot,\cdot)$, that is, 
\begin{displaymath}
M(x,ds_\infty)=\mu(x,ds(0))K(x,s(0),ds(1))K(x,s(1),ds(2))\cdots. 
\end{displaymath}

\begin{definition}[Markov chain Monte Carlo procedure]
If Monte Carlo procedure $\mathcal{M}=(M,e)$ has $M$ as a random Markov measure, we call $\mathcal{M}$
Markov chain Monte Carlo procedure defined on $(X,\mathcal{X},P)$ on $(S,\Theta)$. 
\end{definition}

As a measure of efficiency,  we define consistency and local consistency for a sequence of Monte Carlo procedures. 
Let $(X_n,\mathcal{X}_n,P_n)$ be a probability space, 
$(S_n,\mathcal{S}_n)$ be a measurable space and 
$(\Theta_n,d^n)$ be a complete separable metric space equipped with Borel $\sigma$-algebra $\Xi_n$ for
$n=1,2,\ldots$. Let $\mathcal{M}_n=(M_n,e_n)$ where $e_n=(e_{n,m};m=1,2,\ldots)$ be a Monte Carlo procedure
on $(X_n,\mathcal{X}_n,P_n)$ on $(S_n,\Theta_n)$ for $n=1,2,\ldots$. 

The purpose of the Monte Carlo procedure is to approximate a sequence of probability transition kernels 
$(\Pi_n;n=1,2\ldots)$ from $X_n$ to $\Theta_n$ for each $n=1,2,\ldots$. 
Let $w_n$ be a bounded Lipshitz metric on $(\Theta_n,\Xi_n)$ defined by $d^n$. Then 
\begin{displaymath}
w_n(e_{n,m}(s_m,\cdot),\Pi_n(x_n,\cdot))
\end{displaymath}
measures a loss of the approximation of $\Pi_n(x_n,\cdot)$ by $e_{n,m}(s_m,\cdot)$. 
\begin{displaymath}
W_m(\mathcal{M}_n(x_n,\cdot),\Pi_n(x_n,\cdot)):=\int_{s_m\in S_n^{\mathbf{N}_0}}w_m(e_{n,m}(s_m,\cdot),\Pi_n(x_n,\cdot))M_n(x_n,ds_\infty)
\end{displaymath}
is an average loss with respect to $s_\infty$. We define a risk of the use of the Monte Carlo procedure 
$\mathcal{M}_n$ up to $m$ for an approximation of $\Pi_n$ by
\begin{displaymath}
R_m(\mathcal{M}_n,\Pi_n):=\int_{x_n\in X_n}W_m(\mathcal{M}_n(x_n,\cdot),\Pi_n(x_n,\cdot))P_n(dx_n).
\end{displaymath}

\begin{definition}[Consistency]
A sequence of Monte Carlo procedure $(\mathcal{M}_n;n=1,2,\ldots)$ 
is said to be consistent to $(\Pi_n;n=1,2,\ldots)$ if 
$R_{m_n}(\mathcal{M}_n,\Pi_n)\rightarrow 0$ for any $m_n\rightarrow\infty$. 
\end{definition}

When $\Pi_n(x_n,\cdot)$ tends to a point mass, the above consistency does not provide good information. 
In such a case, we consider local consistency. 
Let $\Theta_n\equiv\Theta\subset\mathbf{R}^p$ and a centering $\hat{\theta}_n:X_n\rightarrow\Theta$ be measurable.
We consider a scaling
$\theta\mapsto n^{1/2}(\theta-\hat{\theta}_n)$.
Let 
\begin{equation}\nonumber
\left\{\begin{array}{c}
\Pi^*_n(x_n,A):=\int 1_A(n^{1/2}(\theta-\hat{\theta}_n))\Pi_n(x_n,d\theta),\\ 
e_{n,m}^*(x_n,s_m,A):=\int 1_A(n^{1/2}(\theta-\hat{\theta}_n))
e_{n,m}(x_n,s_m,d\theta)
\end{array}\right.
\end{equation}
Let $\mathcal{M}_n^*=(M_n,e_n^*)$ for $e_n^*=(e_{n,m}^*;m=1,2,\ldots)$. 

\begin{definition}[Local Consistency]
If $(\mathcal{M}^*_n;n=1,2,\ldots)$ 
is consistent to $(\Pi_n^*;n=1,2\ldots)$, 
$(\mathcal{M}_n;n=1,2,\ldots)$ 
is said to be local consistent to $(\Pi_n;n=1,2,\ldots)$.
\end{definition}

\begin{remark}
This scaling is just one example. For other cases, such as mixture model considered in \cite{Kamatani11}, 
$\theta\mapsto n^{1/2}\epsilon_n^{-1}(\theta-\hat{\theta}_n)$ where $\hat{\theta}_n\equiv 0$ for some $\epsilon_n\rightarrow 0$. 
Moreover, the scaling factor ($n^{1/2}$ or $n^{1/2}\epsilon_n^{-1}$ in the above example) may depend on the observation. 
However, for the current paper, it is sufficient to consider the above scaling $\theta\mapsto n^{1/2}(\theta-\hat{\theta}_n)$.
\end{remark}

In the end of the subsection, 
we briefly review the definition of standard Gibbs sampler and extend it to non-i.i.d. structure. 
Let $(\Theta,d)$ be a complete separable metric space equipped with Borel $\sigma$-algebra $\Xi$.
Let $(X_n,\mathcal{X}_n)$ and $(Y_n,\mathcal{Y}_n)$ be measurable spaces. 
Assume the existence of probability transition kernels
\begin{displaymath}
P_n(d\theta|x_n,y_n),\ 
P_n(dy_n|x_n,\theta),\ 
P_n(d\theta|x_n)
\end{displaymath}
with probability measures 
$
P_n(dx_n,dy_n,d\theta).
$
Assume we have relations
\begin{displaymath}
P_n(dx_n,dy_n,d\theta)=P_n(d\theta|x_n,y_n)P_n(dx_n,dy_n)
=P_n(dy_n|x_n,\theta)P_n(dx_n,d\theta)
\end{displaymath}
where $P_n(dx_n,d\theta)$ and $P_n(dx_n,dy_n)$
are marginal distributions of $P_n(dx_n,dy_n,d\theta)$. 
Moreover, we assume 
\begin{displaymath}
P_n(dx_n,d\theta)=P_n(d\theta|x_n)P_n(dx_n)
\end{displaymath}
where $P_n(dx_n)$ is also a marginal distribution. 
Let 
\begin{displaymath}
\overline{\Pi}_n(x_n,ds)=P_n(dy_n|x_n,\theta)P_n(d\theta|x_n),\ 
\overline{K}_n(x_n,s,ds^*)=P_n(dy_n^*|x_n,\theta)P_n(d\theta^*|x_n,y_n^*)
\end{displaymath}
for $s=(y_n,\theta)$ and $s^*=(y_n^*,\theta^*)$.
Let $\overline{e}_n=(\overline{e}_{n,m};m=1,2,\ldots)$ be 
\begin{displaymath}
\overline{e}_{n,m}(x_n,s_m,A)=m^{-1}\sum_{i=0}^{m-1}1_A(\theta(i))\ (A\in\Xi)
\end{displaymath}
where $s_m=(s(0),\ldots, s(m-1))$ and $s(i)=(y(i),\theta(i))$. 

\begin{definition}[Sequence of standard Gibbs sampler]\label{ssg}
Set $\overline{M}_n$ as a random Markov measure 
generated by $(\overline{\Pi}_n,\overline{K}_n)$. 
Then $(\overline{\mathcal{M}}_n=(\overline{M}_n,\overline{e}_n);n=1,2,\ldots)$
is called 
a sequence of standard Gibbs sampler defined on
$(X_n,\mathcal{X}_n,P_n)$ on $(Y_n\times \Theta,\Theta)$. 
\end{definition}

Using the abbreviation defined in the next subsection, we can write
$\overline{\mathcal{M}}_n=(\overline{M}_n,\theta)$.

\subsection{Abbreviations}
The framework described in the previous subsection is useful as a formal definition for 
Monte Carlo procedures. However, it is sometimes inconvenient 
to write down $e=(e_m;m=1,2,\ldots)$ for every time. 
In this paper we use two abbreviations to denote Monte Carlo procedure $(M,e)$. First one is abbreviation for a class of empirical distribution. All examples of $e$ in the rest  of the aper  has the following form
\begin{displaymath}
e_m(x,s_m,A)=\frac{1}{m}\sum_{i=0}^{m-1}1_A(E(s(i)))
\end{displaymath}
where $s_m=(s(0),\ldots, s(m-1))$ and $E:S\rightarrow \Theta$. 
Then we write $(M,E)$ for $(M,e)$. 
We also use a notation $(M,E(s))$. 
For example, if $s=(y,\theta)$ and $E(s)=\theta\in\Theta$, then 
$(M,E)$ is denoted by $(M,\theta)$. 
If $S=\Theta$ and $E$ is the identity map, we write $(M,\mathrm{id})$. 

The second abbreviation is about transformation. 
Let $F:\Theta\rightarrow\Psi$ where $(\Psi,d_\Psi)$ is a Polish space. 
For a probability transition kernel $\mu(x,d\theta)$, we define
\begin{displaymath}
\mu^F(x,A)=\int_\Theta 1_A(F(\theta))\mu(x,d\theta). 
\end{displaymath}
Similarly, we define
\begin{displaymath}
e_m^F(x,s_m,A)=\int_\Theta 1_A(F(\theta))e_m(x,s_m,d\theta)
\end{displaymath}
for $e=(e_m;m=1,2,\ldots)$. 
Set $e^F=(e_m^F;m=1,2,\ldots)$ and $\mathcal{M}^F=(M,e^F)$. 
Then $\mathcal{M}^F$ is a Monte Carlo procedure defined on 
$(X,\mathcal{X},P)$ on $(S,\Psi)$
and we call $\mathcal{M}^F$ a transform of $\mathcal{M}$. 
For example, if $\mathcal{M}=(M,E)$, then $\mathcal{M}^F=(M,E\circ F)$. 

Now we consider a localization of a transform $\mathcal{M}^F$. 
Let $\Theta\subset\mathbf{R}^p$ and $\Psi\subset\mathbf{R}^q$
be open sets. 
For $\tau\in\Psi$, 
we define a scaling $n^{1/2}(\tau-\hat{\tau}_n)$ where 
$\hat{\tau}_n=F(\hat{\theta}_n)$. 
We write
$\mathcal{M}_n^{F*}=(M_n,e_n^{F*})$ and $\Pi_n^{F*}$ for the scaling of 
$\mathcal{M}_n^F=(M_n,e_n^F)$ and $\Pi_n^F$ with respectively, that is, 
\begin{equation}\label{Fs}
\left\{\begin{array}{c}
\Pi^{F*}_n(x_n,A):=\int 1_A(n^{1/2}(F(\theta)-F(\hat{\theta}_n)))\Pi_n(x_n,d\theta),\\ 
e_{n,m}^{F*}(x_n,s_m,A):=\int 1_A(n^{1/2}(F(\theta)-F(\hat{\theta}_n)))
e_{n,m}(x_n,s_m,d\theta)
\end{array}\right.
\end{equation}
where $e_n^{F*}=(e_{n,m}^{F*};m=1,2,\ldots)$.
We say  
$(\mathcal{M}_n^F;n=1,2,\ldots)$ is locally consistent
to $(\Pi_n^F;n=1,2,\ldots)$ 
if 
$(\mathcal{M}_n^{F*};n=1,2,\ldots)$ is consistent
to $(\Pi_n^{F*};n=1,2,\ldots)$. 
The following lemma states that 
$(\mathcal{M}_n^F;n=1,2,\ldots)$ is locally consistent
if
$(\mathcal{M}_n;n=1,2,\ldots)$ is. 

\begin{lemma}\label{coneq}
Let
$F :\Theta\rightarrow\Psi$ be
$C^1$ map except a compact set $N$ of $\Theta$. 
Assume $P_n(\hat{\theta}_n\in N)\rightarrow 0$ and both the law of $\hat{\theta}_n$
and $\int_{X_n}P_n(dx_n)\Pi_n^*(x_n,\cdot)$ are tight. 
Then $(\mathcal{M}_n^F;n=1,2,\ldots)$ is local consistent to $(\Pi_n^F;n=1,2,\ldots)$ if $(\mathcal{M}_n;n=1,2,\ldots)$ is local consistent to 
$(\Pi_n;n=1,2,\ldots)$.
\end{lemma}

\begin{proof}
Let $B_r(u)=\{v\in\mathbf{R}^p; d(u,v)<r\}$. 
Fix $m_n\rightarrow\infty$
and $r_n\rightarrow0$ such that $r_n n^{1/2}\rightarrow\infty$. 

We first remark that for any $\epsilon>0$, there exists $\delta>0$ such that 
for $N^\delta=\{x\in\Theta; d(x,y)<\delta,y\in N\}$, 
there exists a compact set $K\subset (N^\delta)^c$ such that
\begin{displaymath}
\limsup_{n\rightarrow\infty}P_n(\hat{\theta}_n\in K^c)\le \epsilon. 
\end{displaymath}

Second we consider
\begin{equation}\nonumber
\left\{\begin{array}{c}
\tilde{\Pi}_n(A):=\int_{X_n} P_n(dx_n)\Pi_n(x_n,A),\\ 
\tilde{e}_{n,m}(A):=\int_{X_n,S^{\mathbf{N}_0}} 
P_n(dx_n)M_n(x_n,ds_\infty)e_{n,m}(x_n,s_m,A)
\end{array}\right. .
\end{equation}
Then by assumption, both 
$\tilde{\Pi}_n(B_{r_n}(\hat{\theta}_n)^c)$
and 
$\tilde{e}_{n,m_n}(B_{r_n}(\hat{\theta}_n)^c)$
tends to $0$.

By the differentiability of $F$, we have
\begin{displaymath}
\sup_{u\in K,h\in B_{r_n}(0)}|\frac{F(u+h)-F(u)}{|h|}-\frac{\partial F}{\partial x}(u)^Th|\rightarrow 0
\end{displaymath}
and hence
we can replace $n^{1/2}(F(\theta)-F(\hat{\theta}_n))$
of (\ref{Fs})
by $\partial F(\hat{\theta}_n)^Tn^{1/2}(\theta-\hat{\theta}_n)$
if $\hat{\theta}_n\in K$.

Using this replacement, we have
\begin{displaymath}
W_{m_n}(\mathcal{M}_n^{F*}(x_n),\Pi_n^{F*}(x_n))\le 
o_{P_n}(1)+|\frac{\partial F}{\partial x}(\hat{\theta}_n)|
W_{m_n}(\mathcal{M}_n^*(x_n),\Pi_n^*(x_n))
\end{displaymath}
which means local consistency of $(\mathcal{M}_n^F;n=1,2,\ldots)$
to $(\Pi_n^F;n=1,2,\ldots)$. 
\end{proof}

Roughly speaking, this lemma says that, if 
$\mathcal{M}_n$ is ``equivalent'' to 
$\mathcal{N}_n^F$ for some $F$
and $(\mathcal{N}_n,n=1,2,\ldots)$ is locally consistent, then 
$(\mathcal{M}_n,n=1,2,\ldots)$ is also locally consistent.

We define minimal representation and equivalence of Monte Carlo procedure $(M,E)$. 

\begin{definition}[Minimal representation, equivalence]
Let $(M,E)$ be a Monte Carlo procedure for $E:S\rightarrow\Theta$. 
For a realization $s_\infty=(s(0),s(1),\ldots)$
of $M(x,\cdot)$,
we write 
$M^E(x,\cdot)$, for the law of 
$(E(s(0)),E(s(1)),\ldots)$, that is, 
\begin{displaymath}
M^E(x,A)=\int_{S^{\mathbf{N}_0}}1_A((E(s(0)),E(s(1)),\ldots))M(x,ds_\infty).
\end{displaymath}
 Then we call 
$(M^E,\mathrm{id})$ a minimal representation of $\mathcal{M}=(M,E)$. 
If two Monte Carlo procedures $\mathcal{M},\mathcal{N}$ have the same minimal representation, we call $\mathcal{M},\mathcal{N}$ equivalent. 
\end{definition}

Note that even if 
$\mathcal{M}$ is Markov chain Monte Carlo procedure, 
a minimal representation may lose Markov property of 
original Monte Carlo procedure.

\subsection{Degeneracy and local degeneracy}
We define degeneracy of Monte Carlo procedure. 
Let
\begin{displaymath}
W_m'(\mathcal{M}_n(x_n,\cdot)):=\int_{s_\infty\in S_n^{\mathbf{N}_0}}w_m(e_{n,m}(s_m,\cdot),e_{n,1}(s_1,\cdot))M_n(x_n,ds_\infty),
\end{displaymath}
and
\begin{displaymath}
R_m'(\mathcal{M}_n):=\int_{x_n\in X_n}W_m'(\mathcal{M}_n(x_n,\cdot))P_n(dx_n).
\end{displaymath}

Note that for bounded Lipshitz metric $w$ for probability measures on a metric space $(D,d)$, 
\begin{equation}\label{blfact}
w(\mu,\delta_x)=\int d(x,y)\mu(dy)
\end{equation}
where $\delta_x$ is a Dirac measure. 

\begin{definition}
A sequence of Monte Carlo procedure $(\mathcal{M}_n;n=1,2,\ldots)$ on 
$(X_n,\mathcal{X}_n,P_n)$ on $(S_n,\Theta)$ is said to be degenerate if 
$R_m'(\mathcal{M}_n)\rightarrow 0$ for any $m\in\mathbf{N}$. 
If $(\mathcal{M}_n^*;n=1,2,\ldots)$ is degenerate, we call $(\mathcal{M}_n;n=1,2,\ldots)$ locally degenerate. 
\end{definition}

\begin{remark}
In fact, as a measure of poor behavior, degeneracy is sometimes too wide. 
Roughly speaking, among degenerate Monte Carlo procedures, there are 
relatively good one and bad one. Even if Monte Carlo procedure is
degenerate, sometimes it tends to $\Pi_n$ in a slower rate. This convergence property is called
a weak consistency by \cite{Kamatani11} although the terminology in that paper is slightly different from the current one.  
We can distinguish degenerate Monte Carlo procedures by the rate. 
\end{remark}

The following is an example for non-random Markov chain Monte Carlo procedure. 
Let $B_r(x)=\{y\in\mathbf{R}^p; |x-y|<r\}$. 

\begin{example}
Let $\Theta_n= S_n\equiv\Theta=\mathbf{R}^p$. 
Let $\mathcal{M}_n=(M_n,\mathrm{id})$ be a non-random Markov chain Monte Carlo procedure 
on $\Theta$ where $M_n$ is generated by $(\mu_n, K_n)$.
Let $R_n$ be a probability transition kernel from $\Theta$ to itself and 
$A_n:\Theta\rightarrow (0,1)$ (open interval) be a measurable map. 
Assume $K_n(x,dy)=A_n(x)R_n(x,dy)+(1-A_n(x))\delta_x(dy)$. 

We can show that
 if $(\mu_n;n=1,2,\ldots)$ is tight and (a) if $\sup_{x\in K}A_n(x)\rightarrow 0$
for any compact set $K$, or  (b) if $\sup_{x\in K}R_n(x,B_\epsilon(x)^c)\rightarrow 0$
for any compact set $K$ and $\epsilon>0$, then $(\mathcal{M}_n;n=1,2,\ldots)$ is degenerate. 

To show (a), fix $m\in\mathbf{N}$ and $\epsilon>0$. 
By assumption, there exists a compact set $K$ such that
\begin{displaymath}
\limsup_{n\rightarrow\infty}\mu_n(K^c)\le \epsilon,\ \sup_{x\in K}A_n(x)\rightarrow 0.
\end{displaymath}
Let $E_m:=\{\theta_\infty; \theta(0)=\theta(1)=\cdots=\theta(m-1)\}$
which is an event without any jump in first $m$ steps. 
Then 
\begin{displaymath}
\limsup_{n\rightarrow\infty}M_n(E_m^c)\le 1-\liminf_{n\rightarrow\infty}\int_{x\in K}(1-A_n(x))^m\mu(dx)\le 1-\epsilon. 
\end{displaymath}
On the event $E_m$, $w(e_m(\theta_m),e_1(\theta_1))=0$
where $e_m(\theta_m)=m^{-1}\sum_{i=0}^{m-1}\delta_{\theta(i)}$. 
Hence $\limsup_{n\rightarrow\infty}R'_m(\mathcal{M}_n)\le \epsilon$ which proves 
the first claim. 

To show (b), as above, fix $m\in\mathbf{N}$ and $\epsilon>0$.
Let
$E_m:=\{\theta_\infty; d(\theta(i),\theta(i+1))<\epsilon/2m\ (i=0,\ldots, m-2)\}$
be the event which does not move far from initial point $\theta(0)$ in first 
$m$ steps. 
By assumption, there exists a compact set $K$ such that
\begin{displaymath}
\limsup_{n\rightarrow\infty}\mu_n(K^c)\le \epsilon/2,\ \sup_{x\in K^{\epsilon/2}}R_n(x,B_{\epsilon/2m}(x))\rightarrow 0
\end{displaymath}
where $K^\epsilon=\{x\in\Theta;\exists y\in K,\ \mathrm{s.t.}\ d(x,y)<\epsilon\}$. 
Then $\limsup_{n\rightarrow\infty}M(E_m^c)\le\epsilon/2$ and 
on the event $E_m$, 
\begin{displaymath}
w(e_m(\theta_\infty),e_1(\theta_\infty))\le m^{-1}\sum_{i=1}^{m-1}d(\theta(0),\theta(i))
\le \sum_{i=0}^{m-2}d(\theta(i),\theta(i+1))<\epsilon/2.
\end{displaymath}
Hence $\limsup_{n\rightarrow\infty}R'_m(\mathcal{M}_n)\le \epsilon$ which proves 
the second claim. 
\end{example}

A sequence of consistent Monte Carlo procedures can be degenerate. 
However it is very spacial case. 
In particular, we have the following proposition. We call a sequence of probability transition kernel
$\Pi_n$ from $X_n$ to $\Theta_n$ degenerate if there exists a measurable map
$\hat{\theta}_n:X_n\rightarrow\Theta_n$ such that
\begin{displaymath}
\lim_{n\rightarrow\infty}\int_{x_n\in X_n} w_n(\Pi_n(x_n,\cdot), \delta_{\hat{\theta}_n(x_n)})P_n(dx_n)=0.
\end{displaymath}
If its localization $\Pi_n^*$ is degenerate, we call $\Pi_n$ locally degenerate. 

\begin{proposition}\label{degpro2}
Let $S_n=\Theta_n$ and 
let $(\mathcal{M}_n=(M_n,\mathrm{id});n=1,2,\ldots)$ be consistent to $(\Pi_n;n=1,2,\ldots)$
and also degenerate. Then $(\Pi_n;n=1,2,\ldots)$ is degenerate. 
\end{proposition}

\begin{proof}
By degeneracy, there exists $m_n\rightarrow\infty$ such that 
$R'_{m_n}(\mathcal{M}_n)$ tends to $0$. 
Then
\begin{equation}\label{degenpi}
\int_{x_n\in X_n}\int_{s_\infty\in S_n^{\mathbf{N}_0}} w_n(\Pi_n(x_n,\cdot), \delta_{s(0)})M_n(x_n,ds_\infty)P_n(dx_n)\rightarrow 0
\end{equation}
since the left hand side is bounded by 
$R_{m_n}(\mathcal{M}_n,\Pi_n)+R_{m_n}(\mathcal{M}_n,\delta_{s(0)})$ where both two terms tend to $0$. 
Write marginal distribution of $M_n(x_n,\cdot)$ on $s(0)$ by $\mu_n(x_n,\cdot)$, 
that is, $M_n(x_n,A\times S_n\times S_n\times\cdots)=\mu_n(x_n, A)$. 
Then the above convergence can be rewritten by 
\begin{displaymath}
\int_{x_n\in X_n}\int_{s\in S_n} w_n(\Pi_n(x_n,\cdot), \delta_s)\mu_n(x_n,ds)P_n(dx_n)\rightarrow 0.
\end{displaymath}
By triangular inequality,  $w_n(\delta_s, \delta_t)$ is bounded by 
$w_n(\Pi_n(x_n,\cdot), \delta_s)$ plus $w_n(\Pi_n(x_n,\cdot), \delta_t)$. 
Hence we have
\begin{displaymath}
\int_{x_n\in X_n}\int_{s,t\in S_n} w_n(\delta_s, \delta_t)\mu_n(x_n,ds)\mu_n(x_n,dt)P_n(dx_n)\rightarrow 0.
\end{displaymath}
For each $x_n$, we can find $\hat{\theta}_n(x_n)$ to be 
\begin{displaymath}
\int_{s\in S}w(\delta_s,\delta_{\hat{\theta}_n(x_n)})\mu_n(x_n,ds)\le \int_{s,t\in S_n} w_n(\delta_s, \delta_t)\mu_n(x_n,ds)\mu_n(x_n,dt)
\end{displaymath}
and measurable. 
Therefore we have 
$\int_{x_n}\int_{s\in S}w(\delta_s,\delta_{\hat{\theta}_n(x_n)})\mu_n(x_n,ds)P_n(dx_n)\rightarrow 0$. 
Hence by triangular inequality, we can replace $s(0)$ in (\ref{degenpi}) by $\hat{\theta}_n(x_n)$
which completes the proof. 
\end{proof}

For stationary case, the following proposition is useful to prove degeneracy. 

\begin{proposition}\label{degpro1}
Let  $F_n:S_n\rightarrow\Theta_n$. 
Let $(\mathcal{M}_n=(M_n,F_n);n=1,2,\ldots)$ be stationary. 
Then $(\mathcal{M}_n;n=1,2,\ldots)$ is degenerate if and only if 
\begin{equation}\label{degenstatio}
\int_{x_n\in X_n}\int_{s_\infty\in S^{\mathbf{N}_0}}d^n(F_n(s(0)),F_n(s(1)))M_n(x_n,ds_\infty)P_n(dx_n)
\rightarrow 0. 
\end{equation}
\end{proposition}

\begin{proof}
Since $w_n(\delta_{F_n(s(0))},\delta_{F(s(1))})=d^n(F_n(s(0)),F_n(s(1)))$, 
the sufficiency is obvious by applying $m=2$ for the definition of degeneracy.
On the other hand, if (\ref{degenstatio}) holds, 
take 
$E_m:=\{s_\infty; d^n(F_n(s(i)),F_n(s(i+1)))<\epsilon/2\ (i=0,\ldots, m-2)\}$. 
For fixed $m\in\mathbf{N}$, 
by stationarity, 
$M_n(x_n,E_m^c)\le mM_n(x_n,E_1^c)$ and 
$\int_{x_n\in X_n}M_n(x_n,E_1^c)P_n(dx_n)\rightarrow 0$ by (\ref{degenstatio}). 
By triangular inequality, on the event $E_m$, $d^n(e_{n,m}(s_m), e_{n,1}(s_1))$ 
where $e_{n,m}(s_m)=m^{-1}\sum_{i=0}^{m-1}\delta_{F_n(s(i))}$ is bounded
by 
\begin{displaymath}
m^{-1}\sum_{i=1}^{m-1}d^n(F_n(s(0)),F_n(s(i)))\le\sum_{i=0}^{m-2}
d^n(F_n(s(i)),F_n(s(i+1)))<\epsilon.
\end{displaymath}
Hence $\limsup_{n\rightarrow\infty}R_m'(\mathcal{M}_n)\le\epsilon$ which proves the claim. 
\end{proof}

We consider local degeneracy of the sequence of standard Gibbs sampler defined in Section 6.1 of \cite{Kamatani10}. 

\begin{proposition}\label{onestepGibbs}
A sequence of a standard Gibbs sampler $(\overline{\mathcal{M}}_n;n=1,2,\ldots)$ is degenerate if and only if 
there exists a measurable function 
$\tilde{\theta}_n:X_n\times Y_n\rightarrow\Theta$
such that
\begin{equation}\label{degeq1}
\int w(P_n(d\theta|x_n,y_n),\delta_{\tilde{\theta}_n})P_n(dx_ndy_n)\rightarrow 0.
\end{equation}
Moreover, if $\Theta\subset\mathbf{R}^p$, 
$(\overline{\mathcal{M}}_n;n=1,2,\ldots)$
is locally degenerate under the scaling 
$\theta\mapsto n^{1/2}(\theta-\hat{\theta}_n(x_n))$
if and only if 
there exists $\tilde{\tau}_n:X_n\times Y_n\rightarrow\mathbf{R}^{p}$
\begin{displaymath}
\int w(P_n^*(d\theta|x_n,y_n),\delta_{\tilde{\tau}_n})P_n(dx_ndy_n)\rightarrow 0
\end{displaymath}
where $P_n^*(d\theta|x_n,y_n)$ is the localization of 
$P_n(d\theta|x_n,y_n)$.
\end{proposition}

\begin{proof}
Assume that $(\mathcal{M}_n;n=1,2,\ldots)$ is degenerate. 
Then by Proposition \ref{degpro1} and 
(\ref{blfact}), 
\begin{displaymath}
\int_{x_n,\theta(0),y_n,\theta(1)} d(\theta(0),\theta(1))P_n(d\theta(1)|x_n,y_n)P_n(dy_n|x_n,\theta(0))P_n(d\theta(0)|x_n)P_n(dx_n)
\end{displaymath}
tends to $0$.
Then as in the proof of 
Proposition \ref{degpro2}, there exists a measurable function 
$\tilde{\theta}_n:X_n\times Y_n\rightarrow \Theta$
such that
 \begin{displaymath}
\int_{x_n,y_n,\theta(1)} d(\tilde{\theta}_n(x_n,y_n),\theta(1))P_n(d\theta(1)|x_n,y_n)P_n(dx_ndy_n)
\end{displaymath}
tends to $0$. This proves (\ref{degeq1}) by (\ref{blfact}).

On the other hand, if  (\ref{degeq1}) holds. Then by triangular inequality, 
\begin{displaymath}
w(\delta_{\theta(0)},\delta_{\theta(1)})=d(\theta(0),\theta(1))
\le d(\theta(0),\tilde{\theta}_n(x_n,y_n))+d(\tilde{\theta}_n(x_n,y_n),\theta(1))
\end{displaymath}
and by stationarity, the two terms on the right hand side have the same law. 
We have
\begin{displaymath}
\int_{\theta(0)} d(\theta(0),\tilde{\theta}_n(x_n,y_n))P_n(d\theta(0)|x_n,y_n)=w(P_n(d\theta|x_n,y_n), \delta_{\tilde{\theta}_n(x_n,y_n)}),
\end{displaymath}
and the integral of the right hand side by $P_n(dy_n|x_n)P(dx_n)$
tends to $0$. Hence 
\begin{displaymath}
\int_{x_n,s_\infty}w(\delta_{\theta(0)},\delta_{\theta(1)})\overline{M}_n(x_n,ds_\infty)P_n(dx_n)\rightarrow 0
\end{displaymath}
and degeneracy follows by Proposition \ref{degpro1}. 
The proof for local degeneracy is  the same
replacing sequence $\theta(0),\theta(1),\ldots$ by 
$F_n(x_n,\theta(0)),F_n(x_n,\theta(1)),\ldots$
where $F_n(x_n,\theta)=n^{1/2}(\theta-\hat{\theta}_n(x_n))$. 
\end{proof}

\begin{remark}
By the proposition, it is easy to show that when standard Gibbs sampler is (locally) degenerate, then standard multi-step Gibbs sampler (not defined here) is also (locally) degenerate. This is another validation for the ordering of transition kernels of \cite{LWK}. We could not establish similar relation for local consistency. 
\end{remark}

For local degeneracy, we have the following. We omit the proof since it is the same for 
local consistency. 

\begin{lemma}\label{degeq}
Let
$F :\Theta\rightarrow\Psi$ be
$C^1$ map except a compact set $N$ of $\Theta$. 
Assume $P_n(\hat{\theta}_n\in N)\rightarrow 0$ and  the law of $\hat{\theta}_n$
is tight. 
Then $(\mathcal{M}_n^F;n=1,2,\ldots)$ is local degenerate if $(\mathcal{M}_n;n=1,2,\ldots)$ is local degenerate.
\end{lemma}

\section{Asymptotic properties for cumlative link model}\label{apc}

We consider a cumulative link model. 
Probability space $(X,\mathcal{X},P_X)$ is defined by $X=\mathbf{R}^p$ 
and $\mathcal{X}=\mathcal{B}(\mathbf{R}^p)$
with probability measure $P_X$ having a compact support. 
For $c\ge 2$,  $Y=\{1,2,\ldots, c\}$ and $\mathcal{Y}=2^Y$. 
Let $F$ be a cumulative distribution function on $\mathbf{R}$. 
When $c\ge 3$, a parameter $\theta=(\alpha,\beta)$ is constructed by 
 $\alpha=(\alpha^2,\ldots,\alpha^{c-1})$ 
such that $0<\alpha^2<\cdots<\alpha^{c-1}$ and $\beta\in\mathbf{R}^p$. 
When $c=2$, $\theta=\beta$. 
The model is
\begin{equation}\label{cummod}
x\sim P_X(dx),\ P(y\le j|x)=F(\alpha^j +\beta^T x)\ (j=1,2,\ldots, c)
\end{equation}
with dummy parameters $\alpha^0=-\infty, \alpha^1=0$ and $\alpha^c=+\infty$. 
The parameter space $\Theta\subset\mathbf{R}^{c-2}\times\mathbf{R}^{p}$ is
\begin{equation}
\Theta=\{(\alpha^2,\alpha^3,\ldots,\alpha^{c-1},\beta);0<\alpha^2<\cdots<\alpha^{c-1},\beta\in\mathbf{R}^p\}.
\end{equation}

This cumulative link model is useful for the analysis of ordered categorical data. 
See monographs such as \cite{McCullaghNelder89} and \cite{AA}. 
The analysis for Gibbs sampler for the model will be studied in the next section. 
Before that, in this section, we show the regularity of the model. First we check quadratic mean differentiability.

\subsection{Quadratic mean differentiability of the model}

We recall the definition of quadratic mean differentiability.
Let $(E,\mathcal{E})$ be a measurable space and $\{P(dx|\theta);\theta\in\Theta\}$
be a parametric family on the space where $\Theta\subset\mathbf{R}^p$
be an open set. 
Assume the existence of a $\sigma$-finite measure $\nu$ on $(E,\mathcal{E})$
having $P(dx|\theta)=p(x|\theta)\nu(dx)$ for a $\mathcal{E}$-measurable function 
$p(x|\theta)$ for any fixed $\theta\in\Theta$. 

\begin{definition}[Quadratic mean differentiability]
$P(dx|\theta)$ is called quadratic mean differentiable at $\theta\in\Theta$
if there exists a $\mathbf{R}^p$-valued $\mathcal{E}$-measurable function $\eta(x|\theta)$
such that
\begin{equation}
\int_{X}|\sqrt{p(x|\theta+h)}-\sqrt{p(x|\theta)}-h^T\eta(x|\theta)|^2\nu(dx)=o(|h|^2)
\end{equation}
for any $h\in\mathbf{R}^p$ such that $h\rightarrow 0$. 
\end{definition}

When $P(dx|\theta)$ is quadratic mean differentiable at $\theta\in\Theta$, 
a lot of properties such as local asymptotic normality of the likelihood ratio
hold with minimal assumptions. See monographs such as \cite{LeCamYang00} and \cite{LR}. 

Consider our model (\ref{cummod}). 
The measurable space $(E,\mathcal{E})$ is 
$(X\times Y,\mathcal{X}\otimes\mathcal{Y})$ for our model and 
$\sigma$-finite (in fact, finite) measure is defined by
\begin{displaymath}
\nu(dxdy)=P_X(dx)\sum_{i=1}^c\delta_i(dy).
\end{displaymath} 
For the choice of $\nu$, $p(xy|\theta)$ satisfying $P(dxdy|\theta)=p(xy|\theta)\nu(dxdy)$ is
 $p(xy|\theta)=F(\alpha^y+\beta^Tx)-F(\alpha^{y-1}+\beta^Tx)$. 
We assume the following bit strong regularity condition. 
For $x\sim P_X(dx)$, write the law of $\xi:=(1,x^T)^T$
by $P_\xi$.

\begin{assumption}\label{assqmd}
\begin{enumerate}
\item\label{assqmdfoF} $F(x)=\int_{-\infty}^xf(y)dy$ for a continuous strictly positive measurable function $f$. 
\item The support of $P_\xi$ is compact, which is not included in any subspace of dimension strictly lower than $p+1$. 
\end{enumerate}
\end{assumption}

\begin{proposition}
Under Assumption \ref{assqmd}, $P(dxdy|\theta)$ is quadratic mean differentiable at 
any 
$\theta$. 
\end{proposition}

\begin{proof}
Take $\mathbf{R}^{c+p-2}$-valued measurable function $\eta(xy|\theta)$ to be
\begin{equation}\label{derivxy}
\eta(xy|\theta)=
\frac{\partial_\theta p(xy|\theta)}{2\sqrt{p(xy|\theta)}}
\end{equation}
which is well defined by Assumption \ref{assqmd}
and set $I(\theta)=(I_{ij}(\theta);i,j=1,2,\ldots, p+c-2)$ by
\begin{displaymath}
I(\theta)=4\int\eta(xy|\theta)\eta(xy|\theta)^T\nu(dxdy). 
\end{displaymath}
By Theorem 12.2.2 of \cite{LR}, if $I(\theta)$ is continuous, $P(dxdy|\theta)$
is quadratic mean differentiable. Since $\nu$ is a finite measure, 
it is sufficient to show the existence of $M$ for any 
bounded open set $A$, 
\begin{displaymath}
\sup_{\theta\in A}|\eta(xy|\theta)|\le M\ (x,y\in X\times Y).
\end{displaymath}
Take an open set $A$ to be its closure $\overline{A}\subset\Theta$ is compact. 
Take $\delta>0$ such that
\begin{displaymath}
\delta<\{\alpha^j-\alpha^{j-1};\theta\in A, j=2,\ldots, c-1\}.
\end{displaymath}
Then there exists $M_1$ such that
\begin{displaymath}
\sup_{j=1,2,\ldots, c-1}\{|x|, |\alpha^j+\beta^Tx|;\theta\in A,x\in\mathrm{supp}\ P_X\}\le M_1,
\end{displaymath}
and for the choice of $M_1$, by continuity and positivity of $f$, there exists constants $c_*, c^*\in (0,\infty)$
such that
\begin{displaymath}
c_*<f(x)<c^*\ (x\in [-M_1,M_1]).
\end{displaymath}
Then for $i=2,\ldots, c-1$, 
\begin{displaymath}
F(\alpha^i+\beta^Tx)-F(\alpha^{i-1}+\beta^Tx)
=\int_{\alpha^{j-1}+\beta^Tx}^{\alpha^j+\beta^Tx}f(x)\ge c_*\delta.
\end{displaymath}
For $i=1,c$, choosing $\delta>0$ to be small enough, 
$F(\beta^Tx)\ge F(-M_1)>c_*\delta$ and 
$1-F(\alpha^{c-1}+\beta^Tx)\ge 1-F(M_1)>c_*\delta$ are satisfied. 
Hence the denominator of the right hand side of (\ref{derivxy})
is uniformly bounded for $\theta\in A$. 

For the numerator of (\ref{derivxy}), we have
\begin{displaymath}
\partial_{\alpha^i} p(xy|\theta)=f(\alpha^i+\beta^Tx)1_{\{y=i\}}
-f(\alpha^i+\beta^Tx)1_{\{y=i+1\}}
\end{displaymath}
and 
\begin{displaymath}
\partial_{\beta} p(xy|\theta)=x(f(\alpha^y+\beta^Tx)-f(\alpha^{y-1}+\beta^Tx)).
\end{displaymath}
The absolute values of the above two terms are uniformly bounded by 
$c^*\max\{1,M_1\}$ for $\theta\in A$. 
Hence  (\ref{derivxy}) is uniformly bounded and the claim follows by 
 Theorem 12.2.2 of \cite{LR} by the bounded convergence theorem. 
\end{proof}

For $x_n=(x^1,\ldots, x^n)$ and $y_n=(y^1,\ldots, y^n)$, set
\begin{equation}\label{nor}
Z_n(x_n,y_n|\theta)=n^{-1/2}\sum_{i=1}^n\frac{2\eta(x^iy^i|\theta)}{\sqrt{p(x^iy^i|\theta)}}
\end{equation}
where $\eta(xy|\theta)$ is defined by (\ref{derivxy}). 
This function is called a normalized score function. 
By quadratic mean differentiability, the law of $Z_n(x_n,y_n|\theta)$
tends to $N(0,I(\theta))$. Moreover, if there exists a uniformly consistent, 
the posterior distribution tends to a normal distribution.
In the next subsection, we show the existence of the test. 

\subsection{Uniformly consistent test}\label{uct}

We prepare notations for the large sample setting.
Let 
\begin{displaymath}
(X_n\times Y_n,\mathcal{X}_n\times\mathcal{Y}_n,P_n(dx_ndy_n|\theta))
=
(X\times Y,\mathcal{X}\times\mathcal{Y},P(dxdy|\theta))^{\otimes n}
\end{displaymath}
and write its element $(x_n,y_n)$ where 
$x_n=(x^1,\ldots, x^n)$, $y_n=(y^1,\ldots, y^n)$. 
For $\theta_0\in\Theta$ and $\theta_0\in K\in\Xi$, 
a sequence of measurable functions $\psi_n:X_n\times Y_n\rightarrow [0,1]$
will be called uniformly consistent test for $\theta_0$ against $K^c$ if both
\begin{equation}
\int\psi_n(x_n,y_n)P_n(dx_ndy_n|\theta_0),\ 
\sup_{\theta\in K^c}\int 1-\psi_n(x_n,y_n)P_n(dx_ndy_n|\theta)
\end{equation}
tend to $0$ as $n\rightarrow\infty$. 
We prove the existence of the uniformly consistent test. 
The following lemma states that it is sufficient to construct uniformly consistent test for 
smaller parameter spaces. 

For $\theta_0=(\alpha_0^2,\ldots,\alpha_0^{c-1},\beta_0)$, define
\begin{displaymath}
B_{\epsilon,i}(\theta_0):=\{\theta=(\alpha^2,\ldots,\alpha^{c-1},\beta);(|\alpha^i_0-\alpha^i|^2+|\beta_0-\beta|^2)^{1/2}<\epsilon\}
\end{displaymath}
and $B_\epsilon(\theta_0):=\{\theta;|\theta-\theta_0|<\epsilon\}$. 

\begin{lemma}\label{unilem}
Let $c\ge 3$. 
Suppose that for any $\epsilon>0$ and $i=2,\ldots, c-1$, there exists a uniformly consistent test $(\psi_{n,i};n=1,2,\ldots)$ for $\theta_0$ against 
$B_{\epsilon,i}(\theta_0)^c$. Then for any $\epsilon>0$, there exists a 
uniformly consistent test $(\psi_n;n=1,2,\ldots)$ for $\theta_0$ against $B_\epsilon(\theta_0)^c$. 
\end{lemma}

\begin{proof}
For notational simplicity, set $\theta_0=0$ and write $B_{\epsilon,i}$ and $B_\epsilon$
instead of $B_{\epsilon,i}(\theta_0)$ and $B_\epsilon(\theta_0)$. 
We show that a sequence of test defined by 
$\psi_n:=\max_{i=2,3,\ldots, c-1}\psi_{n,i}$
is a uniformly consistent test for $\theta_0$ against $B_{(c-2)\epsilon}$ if 
$(\psi_{n,i};n=1,2,\ldots)$ is those for $\theta_0$ against $B_{\epsilon,i}$.

First observe that 
\begin{displaymath}
\int\psi_nP_n(dx_ndy_n|\theta_0)\le\sum_{i=2}^{c-1}
\int\psi_{n,i}P_n(dx_ndy_n|\theta_0)\rightarrow 0. 
\end{displaymath} 
On the other hand, by an obvious inequality
\begin{displaymath}
|\theta|^2=\sum_{i=2}^{c-1}|\alpha^i|^2+|\beta|^2\le \sum_{i=2}^{c-1}(|\alpha^i|^2+|\beta|^2), 
\end{displaymath}
for any $\theta\in B_{(c-2)\epsilon}^c$, there exists $i$ such that 
$\theta\in B_{\epsilon,i}^c$. Therefore
\begin{displaymath}
\sup_{\theta\in B_{(c-2)\epsilon}^c}\int (1-\psi_n)P_n(dx_ndy_n|\theta)
\le \max_{i=2,3,\ldots, c-1}\{\sup_{\theta\in B_{\epsilon,i}^c}\int (1-\psi_{n,i})P_n(dx_ndy_n|\theta)\}
\end{displaymath}
which tends to $0$ by assumption. Hence the claim follows. 
\end{proof}

Next we see that for $c=2$ we can construct a uniformly consistent test. 
We use an argument used in Step 1 and 2 of Note 8.4.3 of \cite{LeCamYang00}. 
If we can show the existence of test $\psi_n;X_n\times Y_n\rightarrow [0,1]$
for some $n\in\mathbf{N}$ and a compact set $K$ such that
\begin{equation}\label{c2mod}
\int\psi_n P_n(dx_ndy_n|\theta_0)<\frac{1}{2}<\inf_{\theta\in K^c}
\int\psi_n P_n(dx_ndy_n|\theta),
\end{equation}
then the existence of uniformly consistent test for $\theta_0$ against $B_\epsilon(\theta_0)^c$ follows for any $\epsilon>0$. 
This fact comes from quadratic mean differentiability of the model and 
continuity of $\theta\mapsto P(dxdy|\theta)$ in Proholov metric. 

\begin{lemma}\label{unilem2}
Under Assumption \ref{assqmd} with $c=2$, 
there exists a uniformly consistent test for $\theta_0$ against $B_\epsilon(\theta_0)^c$
for any $\theta_0\in\Theta, \epsilon>0$. 
\end{lemma}

\begin{proof}
We take three steps to construct a uniformly consistent test. In the first step, 
we divide $\Theta$ into $p$ subsets $(\Theta_i;i=1,\ldots, p)$. 
In the second step, we construct a uniformly consistent test $\psi_{n,i}$ for 
each parametric family $\{P(dxdy|\theta);\theta\in\Theta_i\}$. 
In the last step we set $\psi_n=\max_{i=1,2,\ldots, p}\psi_{n,i}$
which will be a uniformly consistent test. 

For the first step, construct $\Theta_i\ (i=1,2,\ldots, p)$. 
Choose $(z_i;i=1,2,\ldots, p)$ from $\mathrm{supp}\ P_X$ to be 
$\mathrm{span}(z_i;i=1,2,\ldots, p)=\mathbf{R}^{p}$. 
Then there exists $\delta>0$ such that for any $\xi\in\mathbf{R}^{p}$
having $|\xi|=1$, there exists $i\in\{1,2,\ldots, p\}$ such that
\begin{displaymath}
|\xi^T z_i|>2\delta.
\end{displaymath}
By $z_i\in\mathrm{supp}\ P_X$, $p_i:=\int_{B_\delta(z_i)}P_X(dx)>0$
and for $|\xi|=1$, there exists $i\in\{1,2,\ldots, p\}$ such that
\begin{displaymath}
\xi^Tx >\delta\ (\forall x\in B_\delta(z_i)), \mathrm{or}\ 
\xi^Tx <-\delta\ (\forall x\in B_\delta(z_i)).
\end{displaymath}
Therefore, if we take 
\begin{displaymath}
\tilde{\Theta}_i:=\{\theta\neq 0; \theta^Tx >\delta|\theta|\ (\forall x\in B_\delta(z_i)), \mathrm{or}\ 
\theta^Tx <-\delta|\theta|\ (\forall x\in B_\delta(z_i))\}
\end{displaymath}
then $\cup_{i=1}^{p}\tilde{\Theta}_i=\mathbf{R}^{p}\backslash\{0\}$. 
To be disjoint, set $\Theta_1=\tilde{\Theta}_1\cup\{0\}$
and $\Theta_i=\tilde{\Theta}_i\backslash\cup_{j=1}^{i-1}\tilde{\Theta}_j$
for $i=1,\ldots, p$. 

In the second step, we construct a uniformly consistent test for the parametric family 
$\{P(dxdy|\theta);\theta\in\Theta_i\}$ for each $i=1,2,\ldots, p$. 
We show that we can construct a test $\psi_{2,i}$ on $X_2\times Y_2$
which satisfies (\ref{c2mod}). Write $x_2=(x^1,x^2)\in X_2$
and $y_2=(y^1,y^2)\in Y_2$. 
The test is
\begin{displaymath}
\psi_{2,i}(x_2,y_2)=
\left\{
\begin{array}{cl}
1/2&\mathrm{if}\ x^1\ \mathrm{or}\ x^2\in B_\delta(z_i)^c,\\
c_i & \mathrm{if}\ x^1,x^2\in B_\delta(z_i),\ \mathrm{and}\ y^1=y^2,\\
0 & \mathrm{otherwise}
\end{array}
\right.
\end{displaymath}
where $c_i\in (0,1)$ will be defined later. Note that since $2(p_i^2+(1-p_i)^2)>1$, 
$\psi_{2,i}:X_2\times Y_2\rightarrow [0,1]$. 
By definition $\int\psi_{2,i}(x_2,y_2)P_2(dx_2dy_2|\theta)$ is
\begin{equation}\label{test}
\frac{1-p_i^2}{2}
+c_i(\int_{B_\delta(z_i)}F(\theta^Tx)P(dx))^2
+c_i(\int_{B_\delta(z_i)}(1-F(\theta^Tx))P(dx))^2.
\end{equation}
When $c_i=1/2$, this value is bounded by
\begin{displaymath}
\frac{1-p_i^2}{2}
+\frac{1}{2}(\int_{B_\delta(z_i)}F(\theta^Tx)P(dx)+\int_{B_\delta(z_i)}(1-F(\theta^Tx))P(dx))^2
\end{displaymath}
which equals to $1/2$. If we take $|\theta|\rightarrow\infty$, by definition of 
$\Theta_i$, $(F(\theta^Tx),1-F(\theta^Tx))$ tends to $(1,0)$ or $(0,1)$
for $x\in B_\delta(z_i)$ and hence (\ref{test}) tends to
\begin{displaymath}
\frac{1-p_i^2}{2}+c_i.
\end{displaymath}
Hence if we take 
$c_i$ slightly larger than $1/2>p_i^2/2$
 to be $\int\psi_{2,i}(x_2,y_2)P_2(dx_2dy_2|\theta_0)<1/2$, there exists a compact 
 set $K$ such that (\ref{c2mod}) holds. Hence we can find a uniformly consistent test for 
 $\theta_0$ against $B_\delta(\theta_0)\cap \Theta_i$. 
 
 In the last step, we take $\psi_n=\max_{i=1,2,\ldots, n}\psi_{n,i}$
 where each $(\psi_{n,i};n=1,2,\ldots)$ is uniformly consistent test for $\theta_0$
 against $B_\delta(\theta_0)\cap \Theta_i$. 
 Then by construction $(\psi_n;n=1,2,\ldots)$ is uniformly consistent test for $\theta_0$
 against $B_\delta(\theta_0)$. 
\end{proof}

Now we extend this test for the model (\ref{cummod}) for $c\ge 3$. 
By Lemma \ref{unilem} it is  sufficient to construct the test for $\theta_0$
against $B_{\epsilon,i}(\theta_0)$
for each $i=2,\ldots, p$. We apply the test constructed in Lemma \ref{unilem2}
for each $i$. Let $Z=\{1,2\}$ and $Z_n=\{1,2\}^n$
and define a map $\pi_i:X\times Y\rightarrow X\times Z$
to be $\pi_i(x,y)=(x,1+1(y>i))$
and $\pi_{n,i}:X_n\times Y_n\rightarrow X_n\times Z_n$ to
be its obvious generalization. 
When $(x,y)\sim P(dxdy|\theta)$, the law of $(x,z)=\pi_i(x,y)$ only depends on 
$\alpha^i$ and $\beta$ defined by 
\begin{displaymath}
x\sim P(dx),\ P(z=1|\alpha^i,\beta, x)=1-P(z=2|\alpha^i,\beta, x)=F(\alpha^i+\beta^Tx).
\end{displaymath}
Therefore it is a model (\ref{cummod}) for $c=2$ with explanatory variable $(1,x^T)^T$
and parameter $(\alpha^i,\beta^T)^T$. 
Write above model by $P(dxdz|\alpha^i,\beta)$. 
For the parametric family $\{P(dxdz|\alpha^i,\beta);\alpha^i\in\mathbf{R},\beta\in\mathbf{R}^p\}$, by Lemma \ref{unilem2}, we can construct a 
uniformly consistent test $(\tilde{\psi}_{n,i};n=1,2,\ldots)$ for $(\alpha_0^i,\beta_0)$ against 
$\{(\alpha^i,\beta); (|\alpha^i-\alpha^i_0|^2+|\beta-\beta_0|^2)^{1/2}\ge \epsilon\}$. 
Then $\psi_{n,i}(x_n,y_n):=\tilde{\psi}_{n,i}(\pi_{n,i}(x_n,y_n))$
defines a uniformly consistent test for $\theta_0$ against $B_{\epsilon,i}(\theta_0)^c$. 
Hence $\psi_n=\max_{i=2,\ldots, c-1}\psi_{n,i}$ is uniformly consistent test for $\theta_0$
against $B_\epsilon(\theta_0)$. As a summary we obtain the following. 

\begin{proposition}\label{unipro}
For the model (\ref{cummod}) under Assumption \ref{assqmd},
there exists a uniformly consistent test for $\theta_0$ against $B_\epsilon(\theta_0)^c$
for any $\theta_0\in\Theta, \epsilon>0$. 
\end{proposition}

If there exists a uniformly consistent test, the posterior distribution has 
consistency under regularity condition on the prior distribution. 
Let $\Lambda(d\theta)=\lambda(\theta)d\theta$ be a prior distribution 
where $d\theta$ denote the Lebesgue measure. 
Let 
\begin{displaymath}
P_n(dx_n,dy_n)=\int_\Theta P_n(dx_ndy_n|\theta)\Lambda(d\theta). 
\end{displaymath}
Assume the existence of 
$P_n(d\theta|x_n,y_n)$ such that
\begin{displaymath}
P_n(d\theta|x_n,y_n)P_n(dx_n,y_n)=P_n(dx_n,dy_n|\theta)\Lambda(d\theta). 
\end{displaymath}
Write $I(\theta)$ for the Fisher information matrix of $P(dxdy|\theta)$ and 
write $\hat{\theta}_n$ for the central value of $P_n(d\theta|x_n,y_n)$. 
The following is a consequence of Bernstein-von Mises's theorem. 
Let $\|\mu-\nu\|=\sup_{A\in\mathcal{E}}|\mu(A)-\nu(A)|$ be the total variation 
distance between probability measures $\mu$ and $\nu$ on $(E,\mathcal{E})$, 

\begin{corollary}\label{bvmthm}
Assume $\lambda$ is continuous and strictly positive. Under Assumption \ref{assqmd}, 
\begin{displaymath}
\int_{X_n,Y_n}\|P_n(d\theta|x_n,y_n)-N(\hat{\theta}_n,n^{-1}I(\hat{\theta}_n))\|P_n(dx_ndy_n)\rightarrow 0. 
\end{displaymath}
\end{corollary}

We will denote $\Pi_n(x_n,y_n,d\theta)$ for $P_n(d\theta|x_n,y_n)$.
We also denote $\Pi_n^*(x_n,y_n,d\theta)$ for its scaling 
by $\theta\rightarrow n^{1/2}(\theta-\hat{\theta}_n)$. 
By the above corollary, 
the total variation distance between $\Pi_n^*(x_n,y_n,d\theta)$
and $N(0,I(\hat{\theta}_n)^{-1})$ tends to $0$. 

\section{Application to Gibbs sampler for cumulative link model}\label{ags}

We consider asymptotic properties of the Gibbs sampler for cumulative 
link model. 
Let $(X\times Y,\mathcal{X}\otimes \mathcal{Y})$
be a probability space defined in Section \ref{apc} and 
let $P(dxdy|\theta)$ be a parametric family defined in (\ref{cummod}). 
Under the same settings as Subsection \ref{uct}, 
we construct Markov chain Monte Carlo methods on the model (\ref{cummod})
and examine its efficiency. 

\subsection{Gibbs sampler and its marginal augmentation}
To construct Gibbs sampler, we introduce a hidden variable $z\in Z=\mathbf{R}$. 
There are several possibilities for the choice of the structure.
We consider two choices among them. 
We refer the former, ``null-conditional update'' and ``$\beta^Tx$-conditional update'' for the latter:
\begin{equation}\label{fullmodel}
\left\{\begin{array}{lll}
x\sim P_X(dx) &z\sim f(z)dz &y=j\ \mathrm{if}\ z\in (\alpha^{j-1}+\beta^Tx,\alpha^j+\beta^Tx]\\
x\sim P_X(dx) &z\sim f(z+\beta^Tx)dz & y=j\ \mathrm{if}\ z\in (\alpha^{j-1},\alpha^j]
\end{array}\right.
\end{equation}
Above update defines $P(dxdydz|\theta)$. For example, for $\beta^Tx$-conditional update
\begin{displaymath}
P(dxdydz|\theta)=\sum_{j=1}^c P_X(dx)f(z+\beta^Tx)1_{(\alpha^{j-1},\alpha^j]}(z)dz\delta_j(dy).
\end{displaymath}
For each construction $P(dxdy|\theta)=\int_{z\in Z}P(dxdydz|\theta)$
is equal to the parametric family defined in (\ref{cummod}). 
As Definition \ref{ssg}, we can construct
a sequence of standard Gibbs sampler 
 $\overline{\mathcal{M}}_n=(\overline{M}_n,\theta)$
on $(X_n\times Y_n,\mathcal{X}_n\otimes\mathcal{Y}_n,P_n(dx_ndy_n))$
on $(S_n,\Theta)$ where 
$S_n=Z^n\times\Theta$.

The Gibbs sampler  $\overline{\mathcal{M}}_n$ is known to work poorly 
except $c=2$ with $\beta^Tx$-conditional update. This 
phenomena can be explained by our approach. 
When $\hat{\theta}_n(x_n,y_n)$ is the central value of $P_n(d\theta|x_n,y_n)$, 
a scaling $\theta\mapsto n^{1/2}(\theta-\hat{\theta}_n(x_n,y_n))$ can be defined. 
We will show that the sequence of a standard Gibbs sampler  $\overline{\mathcal{M}}_n$ 
is not locally consistent because
 the model does not satisfy the regularity condition of 
Theorem 6.4 of \cite{Kamatani10} except the case $c=2$ with $\beta^Tx$-conditional update. The detail will be discussed later.

On the other hand, there are some Markov chain Monte Carlo methods which 
works better than above Gibbs sampler. We consider a marginal augmentation method
introduced by \cite{MengDyk99} (See also closely related algorithm, parameter expansion method by \cite{LiuWu99}). In the method, we introduce a new parameter $g\in (0,\infty)$
with prior $\Lambda_g$
and write $\vartheta=(\theta,g)\in\Theta^X:=\Theta\times (0,\infty)$ for the new parameter set
with new prior distribution
\begin{equation}\label{prior}
\Lambda^X(d\vartheta)=\Lambda(gd\theta)\Lambda_g(dg).
\end{equation}
The new model with new parameter set is defined by
\begin{equation}\label{xfullmodel}
\left\{\begin{array}{lll}
x\sim P_X(dx) &z\sim f(gz)gdz &y=j\ \mathrm{if}\ z\in (\alpha^{j-1}+\beta^Tx,\alpha^j+\beta^Tx]\\
x\sim P_X(dx) &z\sim f(g(z+\beta^Tx))gdz & y=j\ \mathrm{if}\ z\in (\alpha^{j-1},\alpha^j]
\end{array}\right.
\end{equation}
where we refer the former, ``null-conditional update with marginal augmentation'' and ``$\beta^Tx$-conditional update with marginal augmentation'' for the latter. 
Write the above parametric family $P(dxdydz|\vartheta)$. 
The original model $P(dxdydz|\theta)$ corresponds to 
$P(dxdydz|\vartheta=(\theta,1))$. 
Some important properties are summarized as follows: 
\begin{enumerate}
\item
Its $X\times Y$ marginal is written by the original model:
for $\vartheta^*=(\theta^*,g^*)$, 
\begin{displaymath}
\int_{z\in Z}P(dxdydz|\vartheta^*)=P(dxdy|\vartheta^*)=P(dxdy|\theta=g^*\theta^*)
\end{displaymath}
where the parametric family in the right hand side is (\ref{cummod}). 
\item 
The $g\theta$-marginal of prior and posterior distribution for parameter expanded model are the same as those without  expansion, that is
\begin{equation}\nonumber
\left\{\begin{array}{l}
\int_{\vartheta=(\theta,g)\in\Theta^X}1_A(g\theta)\Lambda^X(d\vartheta)=\Lambda(A),\\
\int_{\vartheta=(\theta,g)\in\Theta^X}1_A(g\theta)P_n(d\vartheta|x_n,y_n)=\int_A P_n(d\theta|x_n,y_n). 
\end{array}\right.
\end{equation}
\item The probability distribution $P_n(dx_ndy_n)$ is well defined in the following sense:
\begin{displaymath}
\int_{\Theta^X}P_n(dx_ndy_n|\vartheta)\Lambda^X(d\vartheta)
=\int_{\Theta}P_n(dx_ndy_n|\vartheta)\Lambda(d\theta).
\end{displaymath}
\end{enumerate}
We construct a standard Gibbs sampler $(\overline{M}_n^X,\vartheta)$
on $(X_n\times Y_n,\mathcal{X}_n\otimes\mathcal{Y}_n,P_n(dx_ndy_n))$
 on $(S_n^X,\Theta^X)$ where $S_n^X=Z^n\times \Theta^X$. 

We will call 
$\{(\overline{M}_n^X,g\theta);n=1,2,\ldots\}$ (not $(\overline{M}_n^X,\vartheta)$)
 a sequence of standard Gibbs sampler with marginal augmentation. 
 
 In our approach, we can show a result summarized in Table \ref{nopx}. 
 \begin{table}[htbp]
\begin{center} 
\begin{tabular}{| c | c | c| c | c|} 
\hline & Null & $\beta^Tx$  & Null with MA & $\beta^Tx$ with MA\\ 
\hline $c=2$  & X & O &P&O\\  
\hline $c=3$  & X & X &X&O\\ 
\hline $c\ge 4$  & X & X &X&X\\ \hline 
\end{tabular}
\caption{\label{nopx}Asymptotic properties of Gibbs sampler with and without marginal augmentation (MA). The letter O means local consistency and X means local non-consistency. P means local consistency for $p=1$. }
\end{center}
\end{table} 

According to the table, marginal augmentation has better asymptotic properties
for some cases for $c=2,3$. 
However, for $c\ge 4$, any of Gibbs sampler does not have local consistency 
even with marginal augmentation. 

\begin{remark}
There are some  other Markov chain Monte Carlo methods which improve 
original Gibbs sampler. 
For example, parameter expansion methods are studied in such as \cite{LiuWu99} and \cite{HobertMarchev08}.
These algorithms are closely related to marginal argumentation method
and it seems to have the same asymptotic properties described above. 
In \cite{CowlesSC96}, 
Metropolis-within-Gibbs algorithm is considered. It seems to have 
local consistency even for $c\ge 4$ although the choice of proposal distribution 
is difficult. 
\end{remark}

Figure \ref{Figure2} is the simulation results for cumulative probit model for $c=4$
for Gibbs samplers 
$\beta^Tx$-conditional update with/without marginal augmentation.
These are trajectory of the sequence $\theta(i)\ (i=0,\ldots, m-1)$ for $m=200$ generated by Gibbs samplers. 

\begin{figure}[htbp]
\includegraphics[width=12cm,bb=0 0 779 500]{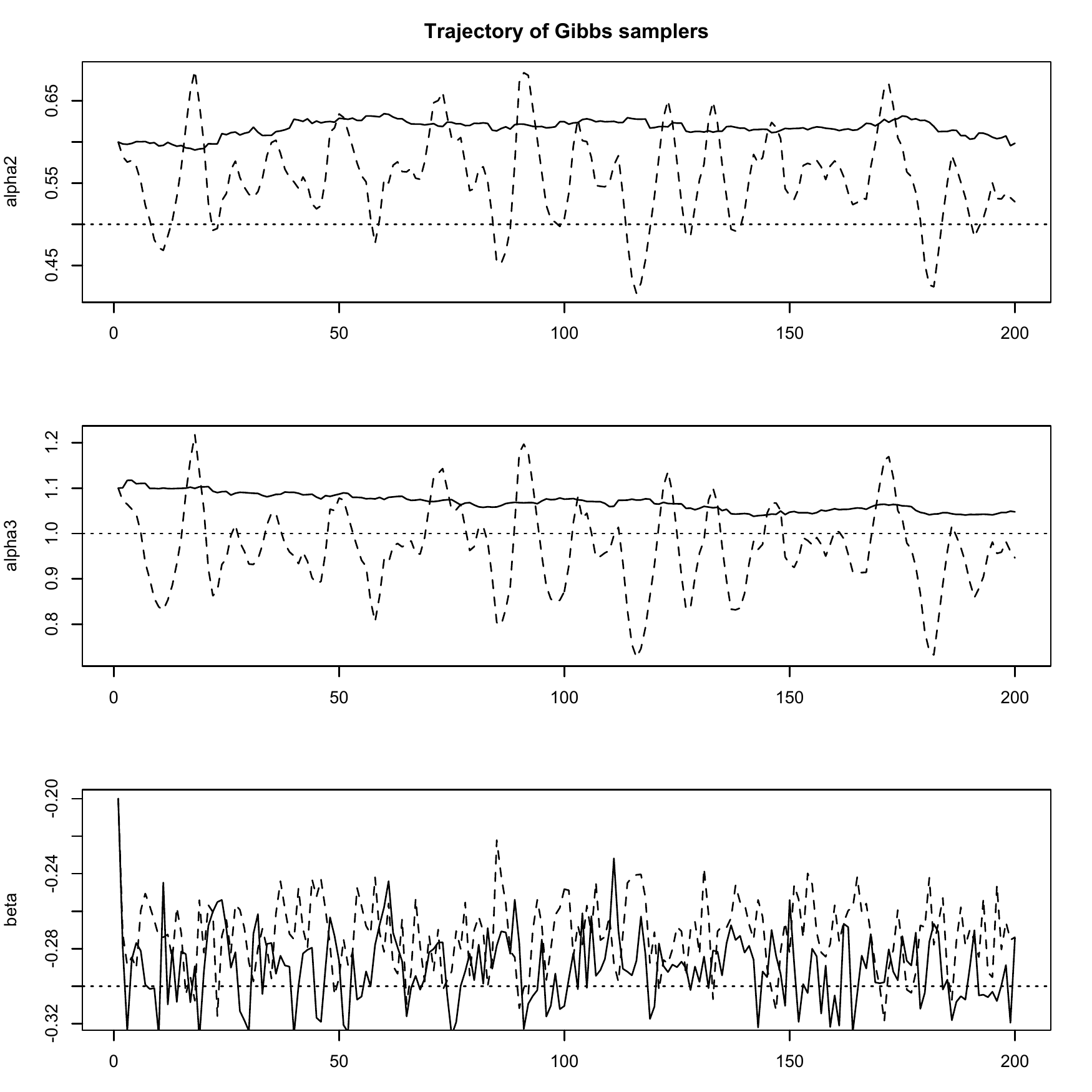}
\caption{Trajectory of the Gibbs samplers for sample size $n=1000$ for $\alpha_2$ (upper) $\alpha_3$ (middle) and $\beta$ (bottom). Solid line is for without MA and dashed lines is for with MA. Horizontal line is the true value. \label{Figure2}}
 \end{figure}
 
 The above figure shows that 
 (a) ``without MA'' is much worse than ``with MA'', 
 (b) for both Gibbs samplers, the mixing property for $\beta$ is not so bad
 and (c) ``with MA'' seems to work well for all parameters. However according to 
 Table \ref{nopx}, ``with MA'' is also locally non-consistent. 
 
 By making a projection $\theta\mapsto \alpha_3/\alpha_2$, 
 we can visualize its local degenerate behavior. 
 Figure \ref{Figure3} is the trajectory of $\alpha_2(i)/\alpha_3(i)\ (i=0,\ldots, m-1)$
 for $m=1000$. Therefore even if ``with MA'' seemed to work well, it has the similar degenerate behavior with ``MA'' and the parameter estimation may cause
 bias. 
 
 \begin{figure}[htbp]
\includegraphics[width=12cm,bb=0 0 779 300]{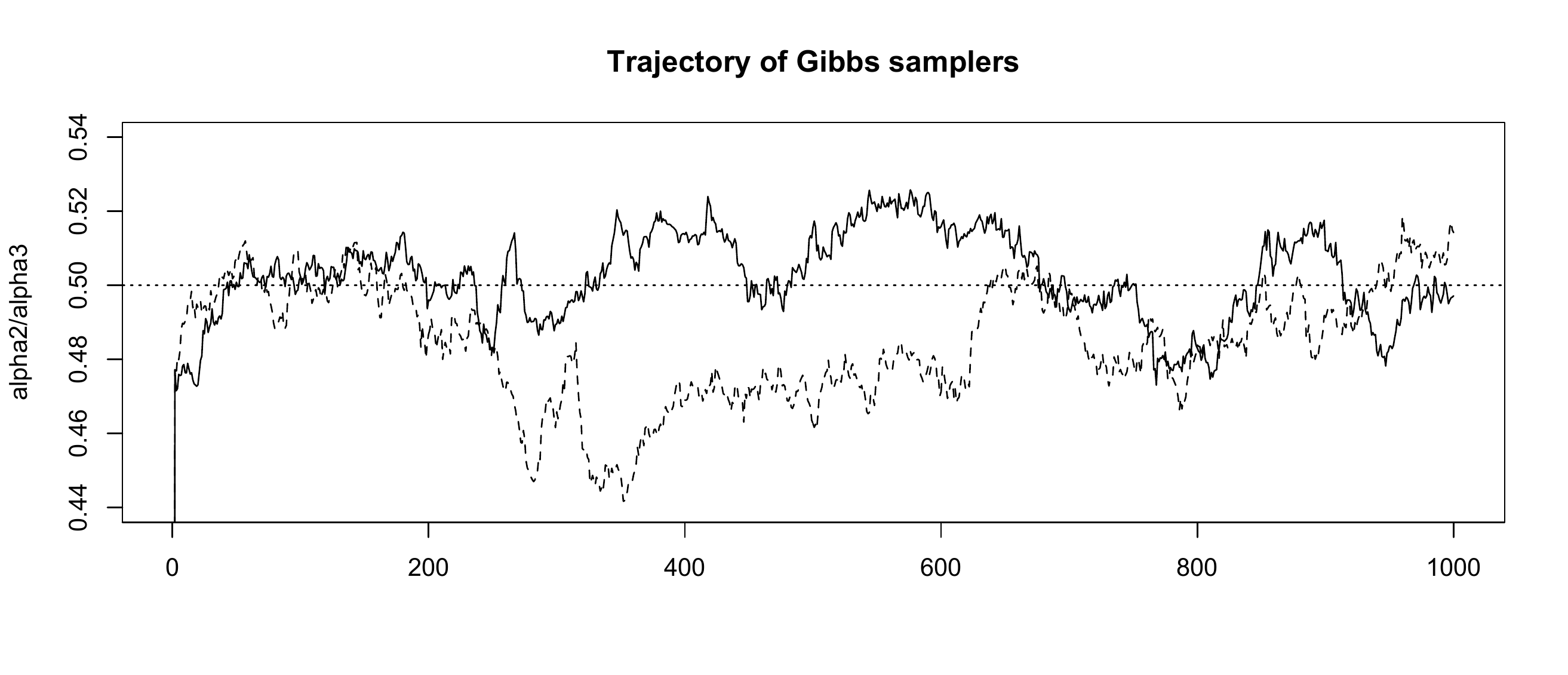}
\caption{Trajectory  for $\alpha_3/\alpha_2$. 
Solid line is for without MA and dashed lines is for with MA. Horizontal line is the true value. \label{Figure3}}
 \end{figure}
 
 In the rest of this section, we prove above results.

\subsection{Approximation of the Gibbs sampler}
We write $(M_n^X,\vartheta)$ for the minimal representation of 
$(\overline{M}_n^X,\vartheta)$. 
Note that the minimal representation of 
$(\overline{M}_n^X,g\theta)$ is $(M_n^X,g\theta)$. 
In this subsection, we construct a formal approximation of $M_n^X$. 
The random Markov measure $M_n^X$ is generated by
$(\Pi_n^X, K^X_n)$ where 
\begin{equation}\nonumber
\left\{\begin{array}{ccc}
\Pi_n^X(x_n,y_n,d\vartheta)&=&P_n(d\vartheta|x_n,y_n)\\
K^X_n(x_n,y_n,\vartheta,d\vartheta^*)&=&
\int_{z_n} P_n(dz_n|x_n,y_n,\vartheta)P_n(d\vartheta^*|x_n,y_n,z_n)
\end{array}\right. 
\end{equation}
where $\vartheta=(\theta,g)$ and $\vartheta^*=(\theta^*,g^*)$. 
Although the parametric family $P(dxdydz|\vartheta)$ does not have 
sufficient regularity described in Theorem 6.4 of \cite{Kamatani10}, 
it has a similar approximation. 

First we remark an important structure of the current model. 
The parameter $\vartheta$ can be divided into $\vartheta_F$ and $\vartheta_M$, where 
the letter ``F'' means ``(Almost) Fixed'' parameter and 
``M'' means ``unfixed (moving)'' parameter. We have 
\begin{displaymath}
P_n(d\vartheta^*|x_n,y_n,z_n)
=P_n(d\vartheta^*_F|x_n,y_n,z_n)
P_n(d\vartheta^*_M|x_n,y_n,z_n).
\end{displaymath}
Depending on the model, $\vartheta_F=\theta, \vartheta_M=g$
 for null-conditional update and $\vartheta_F=\alpha, \vartheta_M=(\beta,g)$ for $\beta^Tx$-conditional update. See the following table. 
We write $\Theta_F^X$ and $\Theta_M^X$ corresponding parameter spaces. 

 \begin{table}[htbp]
\begin{center} 
\begin{tabular}{ c | c | c} 
 & $\vartheta_F$ & $\vartheta_M$  \\ 
\hline Null-conditional   & $\theta$ & $g$\\  
 $\beta^Tx$-conditional  & $\alpha$ & $(\beta,g)$
\end{tabular}
\caption{$\vartheta_F$ and  $\vartheta_M$}
\end{center}
\end{table}

We prepare notation for the a) Fisher information matrix, b) the central value and 
c) the normalized score function
for two models 
A) 
$\{P(dxdy|\vartheta);\vartheta\in\Theta^X\}$
and B)
$\{P(dxdydz|\vartheta);\vartheta_M\in\Theta_M^X\}$ for fixed $\vartheta_F$. 
For fixed $\vartheta\in\Theta^X$, 
we write $X\equiv^a Y$ if 
$X-Y$ tends in $P_n(dx_ndy_n|\vartheta)$-probability to $0$. 

\begin{enumerate}
\item[a)]
Write Fisher information matrices by 
\begin{displaymath}
I(\vartheta) = \left( \begin{matrix} I_{F}(\vartheta)&I_{F,M}(\vartheta)\\ I_{M,F}(\vartheta)&I_{M}(\vartheta) \end{matrix} \right),\ 
K_M(\vartheta)
\end{displaymath}
for models A) and B) with respectively. 
We write $J_M(\vartheta)=K_M(\vartheta)-I_M(\vartheta)$. 
\item[b)]
Write central values by 
\begin{displaymath}
\hat{\vartheta}_n(x_n,y_n) = \left( \begin{matrix} \hat{\vartheta}_{F,n}(x_n,y_n)\\
\hat{\vartheta}_{M,n}(x_n,y_n) \end{matrix} \right),\ 
\hat{\vartheta}_{M,n}(x_n,y_n,\vartheta_F)
\end{displaymath}
for $P_n(d\vartheta|x_n,y_n)$ and $P_n(d\vartheta_M|x_n,y_n,\vartheta_F)$
with respectively. 
Note that $\hat{\vartheta}_{M,n}(x_n,y_n)\equiv^a \hat{\vartheta}_{M,n}(x_n,y_n,\vartheta_F)$. 
We denote $\hat{I}, \hat{J}_M$ and $\hat{K}_M$
for Fisher information matrices 
$I(\vartheta), J_M(\vartheta)$
and $K_M(\vartheta)$
at $\vartheta=\hat{\vartheta}_n(x_n,y_n)$. 
\item[c)]
Write a normalized score function of A) by 
\begin{displaymath}
Z_n(x_n,y_n|\vartheta)=\left(\begin{matrix}Z_{F,n}(x_n,y_n|\vartheta)\\Z_{M,n}(x_n,y_n|\vartheta)\end{matrix}\right).
\end{displaymath}
See (\ref{nor}) for the definition of normalized score function. 
\end{enumerate}

Now we are going to construct an approximation of $K_n^X$. 
Since $\vartheta_F$ is almost fixed parameter
\begin{displaymath}
K^X_n(x_n,y_n,\vartheta,d\vartheta^*)\sim
K^X_{M,n}(x_n,y_n,\vartheta,d\vartheta_M^*)\delta_{\vartheta_F}(d\vartheta_F^*)
\end{displaymath} 
(just a formal sense) where 
\begin{displaymath}
K^X_{M,n}(x_n,y_n,\vartheta,d\vartheta_M^*)=
\int_{z_n} P_n(dz_n|x_n,y_n,\vartheta)P_n(d\vartheta^*_M|x_n,y_n,z_n).
\end{displaymath} 
As an update of $\vartheta_M$, 
$K_{M,n}^X$
is a transition kenel of a standard Gibbs sampler for 
parametric family B). 
With a regularity conditions, we can directly apply Theorem 6.4
to the model B) which yields 
normal approximation of $K_{M,n}^X$
\begin{displaymath}
N(\hat{\vartheta}_{M,n}(x_n,y_n,\vartheta_F)+\hat{K}_M^{-1}\hat{J}_M(\vartheta_M-\hat{\vartheta}_{M,n}(x_n,y_n,\vartheta_F)), 
n^{-1}\hat{K}_M^{-1}+n^{-1}\hat{K}_M^{-1}\hat{J}_M\hat{K}_M^{-1}).
\end{displaymath}
Here we used $(\hat{\vartheta}_{M,n}(x_n,y_n),\vartheta_F)\equiv^a (\hat{\vartheta}_{M,n}(x_n,y_n,\vartheta_F),\vartheta_F)$. 
We denote $\tilde{K}_{M,n}^X$ for this approximated probability transitoin kernel. 
We can rewrite  $\hat{\vartheta}_{M,n}(x_n,y_n,\vartheta_F)$ 
using $\hat{\vartheta}_n(x_n,y_n)$. 
Under $P_n(dx_ndy_n|\vartheta)$, 
\begin{displaymath}
\left\{\begin{array}{lcl}
\hat{\vartheta}_{M,n}(x_n,y_n,\vartheta_F)&\equiv^a&
\vartheta_M+n^{-1/2}\hat{I}_M^{-1}Z_{M,n}(x_n,y_n|\vartheta),\\ 
\hat{\vartheta}_n(x_n,y_n)&\equiv^a&
\vartheta+n^{-1/2}\hat{I}^{-1}Z_n(x_n,y_n|\vartheta).
\end{array}\right. 
\end{displaymath}
Then by a simple algebra, 
\begin{displaymath}
n^{1/2}(\hat{\vartheta}_{M,n}(x_n,y_n,\vartheta_F)
-\hat{\vartheta}_{M,n}(x_n,y_n))\equiv^a
\hat{I}_M^{-1}\hat{I}_{M,F}n^{1/2}(\vartheta_F-\hat{\vartheta}_{F,n}(x_n,y_n)).
\end{displaymath}
This yields an approximation of $K_{M,n}^X$
by normal distribution with mean
\begin{displaymath}
\hat{\vartheta}_{M,n}(x_n,y_n)+\hat{K}_M^{-1}\hat{J}_{M}(\vartheta_M-\hat{\vartheta}_{M,n}(x_n,y_n))
+\hat{K}_M^{-1}\hat{I}_{M,F}(\vartheta_F-\hat{\vartheta}_{F,n}(x_n,y_n))
\end{displaymath}
with variance $n^{-1}\hat{K}_M^{-1}+n^{-1}\hat{K}_M^{-1}\hat{J}_M\hat{K}_M^{-1}$. 
We denote this normal approximation by $\overline{K}_{M,n}^X$. 
Hence we obtain approximation $\overline{K}_{M,n}^X(x_n,y_n,\vartheta,d\vartheta_M^*)\delta_{\vartheta_F}(d\vartheta_F^*)$ of $K_n^X$. 
%

\subsection{Asymptotic properties of the Gibbs sampler}

In this subsection, we study asymptotic properties of 
the Gibbs sampler. It is just a validation of the previous subsection. 
We assume the following.

\begin{assumption}\label{asspri}
\begin{enumerate}
\item
$\Lambda^X$ has the form (\ref{prior}) and 
 $\Lambda(d\theta)=\lambda(\theta)d\theta$
 and 
  $\Lambda_g(dg)=\lambda_g(g)dg$
for Lebesgue measure $d\theta$ and $dg$ where $\lambda(\theta), \lambda_g(g)$ are continuous and 
strictly positive. 
\item $f$ has a derivative $f'$ which is continuous and 
\begin{displaymath}
K:=\int (1+z\frac{f'(z)}{f(z)})^2dz\in (0,\infty). 
\end{displaymath}
\end{enumerate}
\end{assumption}

With the above assumption, we can show that null-conditional update produces 
local degenerate Gibbs sampler
for a map $\theta\rightarrow n^{1/2}(\theta-\hat{\theta}_n(x_n,y_n))$.
For probability transition kernels 
$\mu(x,dy)$ and $K(x,y,dz)$, we denote
\begin{displaymath}
(\mu\otimes K)(x,dy,dz)=\mu(x,dy)K(x,y,dz).
\end{displaymath}

We write $\vartheta=(\alpha^{2},\ldots, \alpha^{c-1},\beta,g)$ and $\vartheta^*=(\alpha^{2*},\ldots, \alpha^{c-1*},\beta^*,g^*)$ for 
elements of $\Theta^X$. We also write 
$(\theta,g)$ or $(\alpha,\beta,g)$ for $\vartheta$ and 
$(\theta^*,g^*)$ or $(\alpha^*,\beta^*,g^*)$ for $\vartheta^*$
with respectively.  

\begin{lemma}\label{nullem}
Under Assumptions \ref{assqmd} and \ref{asspri}, 
for null-conditional update construction with marginal augmentation, the following value tends to $0$:
\begin{equation}\label{nulcon}
\int_{x_n,y_n,\vartheta,\vartheta^*} \min\{\sqrt{n}|\theta-\theta^*|,1\}(\Pi_n^X\otimes K_n^X)(x_n,y_n,d\vartheta,d\vartheta^*)P_n(dx_ndy_n).
\end{equation}
For $\beta^Tx$-conditional update with marginal augmentation, 
\begin{equation}\label{betlemeq1}
\int \min\{\sqrt{n}|\alpha-\alpha^*|,1\}(\Pi_n^X\otimes K_n^X)(x_n,y_n,d\vartheta,d\vartheta^*)P_n(dx_ndy_n).
\end{equation}
tends to $0$.
\end{lemma}

\begin{proof}
We only show the former since proof for the latter is almost the same. 
First we show tightness of $\sqrt{n}(\theta-\theta^*)$. 
We have  $\sqrt{n}(\theta-\theta^*)= \sqrt{n}(\theta-\hat{\theta}_n(x_n,y_n))-\sqrt{n}(\theta^*-\hat{\theta}_n(x_n,y_n))$
and the both terms in the right hand side 
have the same law $\Pi_n^*(x_n,y_n,\cdot)$ defined after Corollary \ref{bvmthm} under 
$(\Pi_n^X\otimes K_n^X)(x_n,y_n,d\vartheta,d\vartheta^*)$. 
Hence the tightness for $\sqrt{n}(\theta-\theta^*)$ follows by Corollary \ref{bvmthm}. 
For any $\epsilon>0$, fix $C_\epsilon$ to be the probability of the event 
$\{\sqrt{n}|\theta-\theta^*|>C_\epsilon\}$ is lower than $\epsilon$ in the limit. 
In the following, we only consider under the event 
$\{\sqrt{n}|\theta-\theta^*|\le C_\epsilon\}$.

As the comment before Proposition \ref{unipro}, we consider 
simpler models. 
It is sufficient to show the convergence 
of $\sqrt{n}|\theta_j-\theta^*_j|$ for
\begin{displaymath}
\theta_j=(\alpha^j,\beta),\ \theta_j^*=(\alpha^{j*},\beta^*)
\end{displaymath}
for $j=2,\ldots, c-1$. For each $j$, for $\xi^i=(1,(x^i)^T)^T$, 
\begin{displaymath}
\left\{
\begin{array}{ccc}
(\theta_j^*)^T\xi^i<z^i&\mathrm{if}& y^i\ge j+1\\
(\theta_j^*)^T\xi^i\ge z^i&\mathrm{if}& y^i\le j
\end{array}
\right. 
\end{displaymath}
since $\vartheta^*$ comes from $P_n(d\vartheta^*|x_n,y_n,z_n)$ (see (\ref{xfullmodel})). 
By simple algebra, for each fixed $\xi_0$, 
\begin{displaymath}
\left\{
\begin{array}{ccc}
(\theta_j^*-\theta_j)^T\xi_0<z^i-\theta_j^T\xi-(\theta_j^*-\theta_j)^T(\xi-\xi_0)&\mathrm{if}& y^i\ge j+1\\
(\theta_j^*-\theta_j)^T\xi_0\ge z^i-\theta_j^T\xi-(\theta_j^*-\theta_j)^T(\xi-\xi_0)&\mathrm{if}& y^i\le j
\end{array}
\right. .
\end{displaymath}

Assume that $\xi_0$ is in the support of $P_\xi$ and set $r=\epsilon/2C_\epsilon$. By the above inequality, we have
\begin{displaymath}
S_n-\frac{\epsilon}{2\sqrt{n}}<(\theta_j^*-\theta_j)^T\xi_0<T_n+\frac{\epsilon}{2\sqrt{n}}
\end{displaymath}
where 
\begin{displaymath}
S_n=\max_{y^i\le j,\xi^i\in B_r(\xi_0)}(z^i-(\theta_j^*)^T\xi^i),\ T_n=\min_{y^i\ge j+1,\xi^i\in B_r(\xi_0)}(z^i-(\theta_j^*)^T\xi^i).
\end{displaymath}
Now we show that the probabilities of events $\{S_n<-\epsilon/2\sqrt{n}\}$
and $\{T_n>\epsilon/2\sqrt{n}\}$ are negligible.
Since the proof is the same, we only show for $S_n$.
The event is
\begin{displaymath}
\{(x_n,y_n,z_n); S_n<-\epsilon/2\sqrt{n}\}=\bigcap_{i=1}^n \{
(x_n,y_n,z_n); (x^i,y^i,z^i)\in E\}
\end{displaymath}
where $E\subset X\times Y\times Z$ is
\begin{displaymath}
E=\{y>j\}\cup\{\xi\notin B_r(\xi_0)\}\cup\{y\le j,\xi\in B_r(\xi_0), (z-(\theta_j^*)^T\xi)<-\epsilon/2\sqrt{n}\}.
\end{displaymath}
Note that $E^c=
\{\xi\in B_r(\xi_0), 0\ge (z-(\theta_j^*)^T\xi)\ge -\epsilon/2\sqrt{n}\}$. 
When we write $p_n$ for the probability of the event $E^c$ with respect to 
$P(dxdydz|\vartheta)$, we have
\begin{equation}\label{pros}
\int 1(\{S_n<-\epsilon/2\sqrt{n}\})P_n(dx_n,dy_ndz_n|\vartheta)
=(1-p_n)^n.
\end{equation}
This value tends to $0$ if $\lim_{n\rightarrow\infty}np_n=+\infty$
and in fact $n^{1/2}p_n$ equals to
\begin{displaymath}
n^{1/2} \int_{\xi\in B_r(\xi_0)}(F(g\theta_j^T\xi)-F(g\theta_j^T\xi-\frac{g\epsilon}{2\sqrt{n}}))P_\xi(d\xi)
\rightarrow \frac{g\epsilon}{2}\int_{\xi\in B_r(\xi_0)}(f(g\theta_j^T\xi)P_\xi(d\xi)
\end{displaymath}
where the limit is strictly positive. Hence (\ref{pros}) tends to $0$ for each $\vartheta$. Its integration by $\Lambda^X$
also tends to $0$ by the bounded convergence theorem.
Hence
$\sqrt{n}(\theta_j^*-\theta_j)^T\xi_0$
tends in probability to $0$. 

By showing the convergence 
$\sqrt{n}(\theta_j^*-\theta_j)^T\xi_i$ for $i=1,2,\ldots, p$
for $\mathrm{span}(\xi_0,\ldots, \xi_p)=\mathrm{supp}\ P_\xi$, 
the claim of the lemma follows. 
\end{proof}

For both cases, if $\{P(dxdydz|\vartheta);\vartheta_M\in\Theta_M^X\}$
for fixed $\vartheta_F\in\Theta_F^X$
has sufficient regularity, then the convergence 
$\int \|(K_{M,n}^X-\overline{K}_{M,n}^X)(x_n,y_n,\vartheta,\cdot)\|P_n(dx_ndy_n|\vartheta)\rightarrow 0$ comes from the proof of \cite{Kamatani10}
as described in the end of the previous subsection. 

Let $\Sigma=\int xx^TP_x(dx)$, $\mu=\int xP_X(dx)$ and 
$L=\int f'(z)/f(z)(f'(z)/f(z)z+1)dz$. 

\begin{lemma}
For each update, $\{P(dxdydz|\vartheta);\vartheta_M\in\Theta_M^X\}$
for fixed $\vartheta_F\in\Theta_F^X$
is quadratic mean differentiable 
having the same support for any $\vartheta_M\in\Theta_M^X$. Moreover there exists a uniformly consistent test.  
In particular, 
\begin{displaymath}
\int \|(K_{M,n}^X-\overline{K}_{M,n}^X)(x_n,y_n,\vartheta,\cdot)\|P_n(dx_ndy_n|\vartheta)\rightarrow 0.
\end{displaymath} 
\end{lemma}

\begin{proof}
For each conditional update, the quadratic mean differentiability of 
the parametric family
 $\{P(dxdydz|\vartheta);\vartheta_M\in\Theta_M^X\}$
for fixed $\vartheta_F\in\Theta_F^X$
comes from
the continuity of the corresponding Fisher information matrices: 
for null-conditional update and $\beta^Tx$-conditional update, the matrices are 
\begin{displaymath}
K_M(\vartheta)=g^{-2}K,\ 
K_M(\vartheta)=
\left( \begin{matrix} g^2K\Sigma&L\mu\\ \mu^TL&g^{-2}K \end{matrix} \right)
\end{displaymath}
with respectively. 
The condition for the support is clear. 
We show the existence of uniformly consistent test. 

For null-conditional update, 
write $\theta_0$ for the fixed parameter $\vartheta_F$. 
Consider a submodel 
$\{P(dxdy|\theta);\theta=g\theta_0,g\in (0,\infty)\}$
of original model. 
Then by Subsection \ref{uct}, this submodel has uniformly consistent test. 
Now we consider a re-parametrization $F: \theta\rightarrow (\theta_0, \theta/|\theta_0|)$. Since $F$ is continuous, 
re-paramezrized model, which is in fact 
$\{P(dxdy|\vartheta);\vartheta=(\theta_0,g), g\in (0,\infty)\}$ has also uniformly consistent test.

 The same argument hold for $\beta^Tx$-conditional update. This proves the claim. 
\end{proof}

We omit the proof of the following proposition since 
it is similar to that of Proposition \ref{MAnc}. 

\begin{proposition}
The standard Gibbs sampler without marginal augmentation is  not locally consistent except for 
 $\beta^Tx$-conditional update for $c= 2$. 
\end{proposition}

For the excepted case, local consistency holds. The proof is directly comes from 
Theorem 6.4 of \cite{Kamatani10}. 

\begin{proposition}
The standard Gibbs sampler without marginal augmentation is  locally consistent for 
 $\beta^Tx$-conditional update for $c= 2$. 
\end{proposition}

\begin{proof}
In this case, the regularity condition of Theorem 6.4 of \cite{Kamatani10}
is satisfied. Hence the claim holds. 
\end{proof}

\begin{proposition}\label{MAnc}
The standard Gibbs sampler with marginal augmentation is not locally consistent in the following cases:
\begin{enumerate}
\item $p\ge 2$ or $c\ge 3$ for null-conditional update. 
\item  $c\ge 4$ for $\beta^Tx$-conditional update. 
\end{enumerate}
\end{proposition}

\begin{proof}
For the null-conditional update wtih marginal augmentation, 
$\mathcal{M}_1=(M_n^X,\theta)$ is locally degenerate
by Lemma \ref{nullem} and Proposition \ref{degpro1}.
Therefore, by $F(\theta)=\theta/|\theta|$, 
$\mathcal{M}_2=\mathcal{M}_1^F=(M_n^X,\theta/|\theta|)$
is locally degenerate by Lemma \ref{degeq}. 
On the other hand, if $\mathcal{M}_3=(M_n^X,g\theta)$ is locally consistent, 
by mapping $G(\theta)=\theta/|\theta|$, 
$\mathcal{M}_2=\mathcal{M}_3^G$ should be locally consistent by Lemma \ref{coneq}. 
Since $P_n(d\vartheta|x_n,y_n)$ is not degenerate with the scaling with map $F$
for $p\ge 2$ or $c\ge 3$, 
it is impossible by Proposition \ref{degpro2}. Hence $\mathcal{M}_3=(M_n^X,g\theta)$ is not locally consistent. 

It is quite similar for $\beta^T x$-conditional update.
For this case,  $\mathcal{M}_1=(M_n^X,\alpha)$ is locally degenerate
and hence $\mathcal{M}_2=(M_n^X,\alpha/|\alpha|)$ is also locally degenerate
by a map $F(\alpha)=\alpha/|\alpha|$.  
On the other hand, 
if $\mathcal{M}_3=(M_n^X,g\theta)$ is locally consistent, 
then $\mathcal{M}_2=(M^X_n,\alpha/|\alpha|)$ should be locally consistent 
since $\mathcal{M}_2=\mathcal{M}_3^G$ for a map $G(\alpha,\beta)=\alpha/|\alpha|$. 
Since $P_n(d\vartheta|x_n,y_n)$ is not degenerate with the scaling with map $F$
for $c\ge 4$, 
it is impossible. 
Hence $\mathcal{M}_3=(M_n^X,g\theta)$ is not locally consistent.
\end{proof}

\begin{proposition}
The standard Gibbs sampler with marginal augmentation is  locally consistent in the following cases:
\begin{enumerate}
\item Null-conditional update for $c= 2$ and $p=1$. 
\item  $\beta^Tx$-conditional update for $c=2, 3$. 
\end{enumerate}
\end{proposition}

\begin{proof}
For null conditional update case, consider 
$\mathcal{M}_1=(M_n^X,g\theta)$ where $\theta=\beta$. 
The probability transition kernel $K^X_n(x_n,y_n,\theta,d\theta^{**})$ of its minimal representation is
\begin{displaymath}
\int_{z_n,g^*,\theta^*} P_n(dz_n|x_n,y_n,(\theta,1))P_n(dg^*d\theta^*|x_n,y_n,z_n)\delta_{g^*\theta^*}(d\theta^{**}).
\end{displaymath} 
We show that we can replace $\theta^*$ by $\theta$ in the above transition kernel. 
Write $\hat{g}_n(z_n)$ for the central value of $P_n(dg|x_n,y_nz_n)=P_n(dg|z_n)$. 
First we apply Bernstein von-Mises's theorem for $\{P(dxdydz|\vartheta);\vartheta=(1,g), g\in (0,\infty)\}$
for the approximation $P_n(dg|x_n,y_n,z_n)\sim N(\hat{g}_n(z_n), n^{-1}K^{-1})$. By this approximation,  we can approximate $K_n^X$
by $L^X_n(x_n,y_n,\theta,d\theta^{*})$ defined by
\begin{displaymath}
\int_{z_n,\theta^*} P_n(dz_n|x_n,y_n,(\theta,1))\phi(\theta^{**};\theta^*\hat{g}_n(z_n),n^{-1}(\theta^*)^2K^{-1})P_n(d\theta^*|x_n,y_n,z_n).
\end{displaymath} 
For some continuous function $C$, uniformly in $g^*$, 
\begin{displaymath}
|\phi(\theta^{**};\theta^*\hat{g}_n(z_n),n^{-1}(\theta^*)^2K^{-1})
-\phi(\theta^{**};\theta\hat{g}_n(z_n),n^{-1}\theta^2 K^{-1})|\le n^{1/2}|\theta^*-\theta| C(\hat{g}_n(z_n)).
\end{displaymath} 
Hence by tightness of $\hat{g}_n(z_n)$ and convergence of
$n^{1/2}|\theta^*-\theta|$ to $0$ in probability, we can replace $\theta^*$
of $L_n^X(x_n,y_n,\theta,d\theta^{**})$ by $\theta$ (see the proof of Theorem 6.4 of \cite{Kamatani10}). Then using Bernstein-von Mises's theorem again, 
it is validated to replace $\theta^*$ in $K_n^X$ in the sense of 
 $\int\|(K_n^X-{K}_n^{'X})(x_n,y_n,\theta,\cdot)\|P_n(dx_ndy_n|\theta)\rightarrow 0$
where ${K}_n^{'X}$ is the transition kernel after replacement of $\theta^*$ by $\theta$. 
We already have an approximation of $K^{'X}_n$. 
Therefore $\mathcal{M}_1$ is locally consistent by the convergence of total variation by the same argument in the proof of Theorem 6.4 of \cite{Kamatani10}. 

By the similar argument, for $\beta^Tx$-conditional update for $c=2, 3$, 
$(M_n^X,\vartheta)$ or $(M_n^X,(\beta/\alpha,g\alpha))$
are locally consistent with respectively. 
Therefore 
$(M_n^X,g\theta)$
is locally consistent
by a map $F(\vartheta)=g\theta$ for the former and $F(\beta,g)=(g\beta,g)$
for the latter. 
\end{proof}

\bibliographystyle{plainnat}

\end{document}